%% file: Full_main10_pages.tex
\documentclass[10pt,twocolumn,twoside]{ieeeconf}

\usepackage{subcaption}
\usepackage[font=small]{caption}
\usepackage{afterpage}
\usepackage{bbold}
\usepackage{comment}
\usepackage{epstopdf}
\usepackage{subcaption}

\usepackage{color}
\usepackage{todonotes}
\usepackage{mathrsfs}
\newtheorem{remark}{\bf{Remark}}

\newcommand{\bs}{\boldsymbol}

\IEEEoverridecommandlockouts                              
\usepackage{epsfig} 
\usepackage{mathptmx} 
\usepackage{times} 
\usepackage{amsmath} 
\usepackage{amssymb}  
\usepackage{multirow}
\usepackage{array}
\usepackage{mathtools}
\usepackage{tabularx}
\usepackage{cite}
\usepackage{footnote}
\makesavenoteenv{tabular} 
\usepackage{gensymb}
\usepackage{bm} 

\usetikzlibrary{backgrounds}
\usepackage{amsmath,amssymb}
\usepackage{bm}
\usepackage{tikz}
\usetikzlibrary{arrows.meta,calc}

\usepackage{algorithm,algorithmic}
\usepackage{graphicx}
\usepackage{booktabs}
\floatname{algorithm}{Method}
\usepackage{color,soul}
\usepackage{comment}
\usepackage{relsize}
\usepackage{tikz}
\usepackage{minted}
\usepackage{xcolor}
\usepackage{listings}

\usepackage{xparse}

\NewDocumentCommand{\codeword}{v}{%
\texttt{\textcolor{black}{#1}}%
}

\newtheorem{definition}{\textbf{Definition}}

\newtheorem{Lemma}{\bf{Lemma}}
\newtheorem{Assumption}{\bf{Assumption}}
\newtheorem{Theorem}{\bf{Theorem}}

\makeatletter
\newcommand{\definetrim}[2]{%
  \define@key{Gin}{#1}[]{\setkeys{Gin}{trim=#2,clip}}%
}
\makeatother

\makeatletter
\let\NAT@parse\undefined
\makeatother
\usepackage[hidelinks]{hyperref}    

\title{\LARGE \bf
A Distributed Framework for Data-Driven Safe Coordination in Leader–Follower Networks}

\author{Mirhan Urkmez, Maryam Sharifi and Shahab Heshmati-Alamdari~\IEEEmembership{Member,~IEEE} 
\thanks{ Mirhan Urkmez and Shahab Heshmati-alamdari are with the Section of Automation \& Control, Department of Electronic Systems, Aalborg University, Denmark, Email: {\tt\small \{mu@es.aau.dk;shhe@es.aau.dk\}}. Maryam Sharifi is with ABB Corporate Research, Västerås, Sweden, Email: {\tt\small  \{maryam.sharifi@se.abb.com\}}. 
}
}

\begin{document}

\maketitle

\begin{abstract}
This paper addresses connectivity preservation in leader–follower multi-agent systems with unknown control-affine dynamics and local state information. We introduce the distributed data-driven zeroing control barrier function (3D-ZCBF) framework, which ensures the controlled invariance of safety sets by identifying derivative bounds from input–state data without requiring explicit models of high-dimensional agent dynamics. In this work, we derive the explicit, decoupled safety conditions necessary to maintain connectivity for leader–leader, and follower–follower pairings. These individual constraints, along with the leader-follower conditions, are aggregated into explicit system-wide conditions that formally guarantee the preservation of the entire communication network. Furthermore, we provide a quantitative analysis demonstrating how the size of the collected data set and the accuracy of the learned Jacobian bounds impact the feasibility of the safety certificates. The proposed conditions are implemented via a projection-based controller, and simulations confirm that these explicit 3D-ZCBF requirements effectively maintain system-level connectivity using only local, two-hop information.
\end{abstract}

\section{Introduction}\label{sec:introduction}
Multi-agent systems (MAS) are increasingly being studied and used in diverse autonomous applications, including cooperative swarm robotics and connected vehicle environments, where collective behavior is essential for mission success \cite{olfati2007consensus, brambilla2013swarm}. In these domains, agents must achieve coordinated objectives, e.g.,  formation keeping \cite{egerstedt2001formation}, coverage \cite{atincc2014supervised}, or target tracking \cite{xia2021multi}, while ensuring that fundamental safety constraints such as collision avoidance and connectivity maintenance are never violated. Among various coordination architectures, leader–follower structures are particularly common, where leaders guide the collective motion of followers through local interactions and allow for implementing a scalable framework \cite{hu2010distributed}, \cite{roque2020decentralized}. The distributed nature of 
such systems, coupled with communication delays, model uncertainty, and dynamic environments, makes 
design of safe control strategies inherently challenging, as even a single local violation can precipitate cascading network-level failures\cite{wang2017safety}. {This challenge is further compounded when the agent or leader–follower dynamics are unknown or only partially known, since classical model-based control and safety certification methods cannot be directly applied. At the same time, agents typically have access only to local neighbor information, which further increases the challenge for designing safe controllers. There are some previous work that consider several effects of uncertainties such as external disturbances \cite{cao2015leader}, model uncertainties \cite{zhang2020robust}, unknown communication delays \cite{tian2012high}. However, they intrinsically rely on a dynamic model structure of the agents which is affected by an unknown parameter.
In this work, we focus on connectivity maintenance in unknown leader–follower multi-agent systems, where each agent has only local state information and inputs must be computed in a distributed manner. Connectivity preservation \cite{griparic2022consensus} is a fundamental safety requirement that ensures reliable coordination and information sharing among agents. It requires that the communication network remain connected over time; that is, all active communication links between agents are maintained. Equivalently, inter-agent distances must remain within the maximum communication range, allowing each agent to continuously share information with its neighbors and adhere to coordinated commands. Ensuring this in practice calls for systematic safety-oriented control strategies. }

A promising approach to address safety-critical requirement for autonomous systems, including those requiring connectivity preservation, is the use of Control Barrier Functions (CBFs), which define the safe set as the nonnegative level set of a continuously differentiable function and provides conditions for the forward invariance enforcement of the safe set \cite{ames2016control}. Since their introduction, CBFs have been widely applied across a range of domains, including autonomous systems \cite{fisac2018general,heshmati2024control}, aerial  \cite{huang2021homography,zhou2022iblf}, and ground robots \cite{wang2016multi,wu2022continuous}, as they provide a computationally tractable framework with formal safety guarantees that minimally alters nominal control laws. In multi-agent contexts, distributed formula+tions of CBFs enable each agent to enforce local safety constraints while cooperating toward common goals \cite{tan2021distributed}. As previously mentioned, leader–follower architectures are a prominent paradigm for MAS coordination. They simplify cooperative control by delegating global mission planning to a small set of leaders while followers coordinate through local interactions, which improves scalability and reduces communication load \cite{Maryam_TCNS}. However, these architectures also introduce safety challenges: followers often have only partial information about leader and neighbor dynamics, agents may be heterogeneous, and safety must be preserved under sensing noise, delays, and communication losses \cite{panagou2015distributed}. Robust safety in these conditions calls for control methods that remain effective with uncertain or incomplete neighbor dynamics \cite{usevitch2019resilient}.

Recent research has begun addressing this challenge through robust and data-driven extensions of CBFs. A robust formulation have been developed to account for modeling errors and interaction uncertainty in multi-agent settings by learning uncertainty bounds \cite{cheng2020safe}. A similar work \cite{kim2025robust} provided theoretical guarantees of safety under limited model knowledge for a high relative degree system with and unknown, dynamically evolving obstacles. Complementary works have explored data-driven and learning-based formulations of CBFs. These include hybrid CBFs learned directly from data or demonstrations \cite{lindemann2021learning}, as well as graph-based neural CBF representations that enable distributed collision avoidance using only local observations in large-scale multi-agent systems \cite{zhang2023neural}. Measurement uncertainty in dynamic environments has recently been addressed through distributionally robust CBF constructions \cite{long2025sensor}.

Most CBF-based approaches for multi-agent systems assume that each agent has access to either a model of its neighbors or reliable estimates of their states. In real environments, however, agents often receive only raw measurements such as relative positions, velocities, or communicated control inputs, without knowing the underlying dynamics that generate them. Under these limited sensing and communication conditions, enforcing interaction-level properties becomes difficult. This is especially true for connectivity enforcement, which is essential for coordination and information exchange. These limitations point to the need for safety mechanisms that can operate directly on observed data while retaining the safety guarantees associated with CBF-based control, particularly in leader–follower settings with uncertain or unknown neighbor behavior.
\subsection{Contribution}

To address connectivity maintenance in unknown leader–follower multi-agent systems, we introduce the distributed data-driven zeroing CBF (3D-ZCBF)
framework. This framework provides the basis for controlled invariance of safety sets by identifying bounds on CBF time derivatives directly from collected input–state data, rather than requiring an explicit model of the high-dimensional agent and neighbor dynamics. By learning these derivative bounds, the method provides a way to enforce safety constraints in systems where the underlying dynamics are unknown. In this work, we provide the explicit conditions required to maintain connectivity across the network. We establish the precise requirements for preserving links between different agent pairs and aggregate these into a unified, system-wide safety guarantee. Additionally, we provide results showing how the size of the data set and the tightness of the learned bounds impact the overall results.

We summarize our main contributions as follows:
\begin{itemize}

\item Derivation of explicit, decoupled 3D-ZCBF conditions for all inter-agent edge types, ensuring that connectivity constraints are locally enforceable by individual agents.

\item Formulation of system-level conditions for connectivity preservation that aggregate all individual distributed constraints into a collective safety guarantee for the entire network.

\item A quantitative analysis of how the size of the input–state data set and the tightness of the learned Jacobian bounds affect the resulting feasibility of the safety certificates.

\end{itemize}

\section{Preliminaries}
\subsection{Notations}
The set of integers up to $n$ is denoted by $\mathbb{Z}^+_n$. A continuous function $\alpha :[0,a)\to \mathbb{R}_+$ with $a\in\mathbb{R}_+$ is a class $\mathcal{K}$ function if it is strictly monotonically increasing and $\alpha(0)=0$, and is of class $\mathcal{K}_\infty$ if it is a class $\mathcal{K}$ function with $a=\infty$ and $\lim_{x\to\infty}\alpha(x)=\infty$.
For a vector/matrix $P$, $P^+\triangleq \max\{P,0\}$ and $P^-\triangleq P^+ - P$ with element-wise maximum. Moreover, vector inequalities are element-wise. Given a set of vectors $v_i \in \mathbb{R}^n, i\in I$, the notation $(v_i)_{i \in I}$ denotes the \emph{stacked vector} obtained by concatenating the individual vectors $v_i$ for all $i \in I$ in increasing order of their indices and $|I|$ is the cardinality of set $I$.

\subsection{Communication Graph}
A communication graph $G$ is denoted by $G = ( V, E)$, where $V = \{ 0,...,M-1\} $ is a finite nonempty node set, and $E \subseteq V \times V$ is the edge set that represents the communication links between the nodes, with $M$ being the number of nodes. The subsets of nodes corresponding to leaders and followers are denoted by $V_l \subseteq V$ and $V_f \subseteq V$, respectively.
The sets of edges connecting leaders to leaders, leaders to followers, and followers to followers are denoted by $E_{ll} \subseteq E$, $E_{lf} \subseteq E$, and $E_{ff} \subseteq E$, respectively, where $E_{ll} = \{ e_{ij} \in E \mid i,j \in V_l \}$, $E_{lf} = \{ e_{ij} \in E \mid i \in V_l, j \in V_f \}$, and $E_{ff} = \{ e_{ij} \in E \mid i,j \in V_f \}$, where $e_{ij} \in E$ denotes an edge from node $i$ to node $j$.
The edge $e_{ij}\in E$ shows that node $i$ obtains information from node $j$ and vice versa. In other words, agents $i$ and $j$ are neighbors. The neighbor set of node $i$ is defined as ${\mathcal{N}_i} = \{ j|e_{ij} \in E\} $.   
    The set of neighbors shared by two agents $k$ and $j$ is defined as
    $\mathcal{N}_{kj} := \mathcal{N}_k \cap \mathcal{N}_j$.
    The set of neighbors unique to agent $k$, excluding $j$, is defined as
    $\mathcal{N}_{k \setminus j} := \mathcal{N}_k \setminus (\mathcal{N}_{kj} \cup \{j\})$.

\subsection{Control Barrier Functions}\label{sec:cbf}
\begin{definition}[Zeroing Control Barrier Function \cite{ames2016control}]\label{def:ZCBF}
Consider the control-affine system
\begin{equation}\label{sysZCBF}
    \dot{\bs x} = f(\bs x) + g(\bs x)\,\bs u,
\end{equation}
where $\bs x \in X \subset \mathbb{R}^n$ is the state, $\bs u \in U \subset \mathbb{R}^m$ is the control input, the set $U$ is compact, and functions $f(\cdot)$ and $g(\cdot)$ are locally Lipschitz. Let $h: X \subset \mathbb{R}^n \rightarrow \mathbb{R}$ be continuously differentiable function and define the safe set $S\triangleq \{\bs x \in X \mid h(\bs x) \ge 0\}$. The function $h$ is a zeroing control barrier function (ZCBF) on $X$ if there exists an extended class-$\mathcal{K}_\infty$ function $\alpha$ such that
\begin{equation}\label{eq:ZCBFcond}
    \sup_{\bs u \in U} \Big\{ \nabla h(\bs{x})^{\top}(f(\bs{x})+g(\bs{x})\bs{u}) + \alpha(h(\bs{x})) \Big\} \ge 0,
    \quad \forall\, \bs x \in X.
\end{equation}
\end{definition}

\section{Problem Formulation}
Consider a leader-follower multi-agent system with $M$ agents in an $n$-dimensional environment. Let the state and input vectors of all agents be, respectively, stacked as
\begin{align}  \label{eq:stacked}
\bs{x} = (\bs{x}_i)_{i \in V} \!\in \!X \!\subset\! \mathbb{R}^{nM}\!\!\!\!\!,\! \quad\!\!
\bs{u}\! = \!(\bs{u}_i)_{i \in V}\! \in \!U \!\subset \!\mathbb{R}^{mM}\!\!\!\!,\!\!
\end{align}
where $\bs{x}_i \in  \mathbb{R}^n$ and $\bs{u}_i \in  \mathbb{R}^m$ for all $i \in V$ and state and input sets $X, U$ are bounded. For a given state vector $\bs{x} \in X$ and a subset of agents $I \subseteq \{1, \dots, M\}$, define the stacked state vector of the agents in $I$ as
\begin{align*}
    &\bs{x}_I := (\bs{x}_i)_{i \in I} \in \mathbb{R}^{n|I|}.
\end{align*}
Similarly, define the projection of $X \subseteq \mathbb{R}^{nM}$ onto the states of the agents in $I$ as
\begin{align*}
&X_I := \{ \bs{x}_I \mid \bs{x} \in X \} \subseteq \mathbb{R}^{n|I|}.
\end{align*}
When $I=\{i\} $ for some $i\in V $, we simply denote $X_I$ with $X_i$. The same stacking and projection notation applies to the inputs, i.e.,
\begin{align*}
\bs{u}_I := (\bs{u}_i)_{i \in I} \in \mathbb{R}^{m|I|}, \quad
U_I := \{ \bs{u}_I \mid \bs{u} \in U \} \subseteq \mathbb{R}^{m|I|}.
\end{align*}


The leader-follower multi-agent dynamics considered in this work can be written as
\begin{subequations}\label{agentZCBF}
    \begin{align}
    \dot{\bs{x}}_j &= f_j(\bs{x}_j, \bs{x}_{\mathcal{N}_j}) + g_j(\bs{x}_j)\bs{u}_j, \quad j  \in V_l, \label{agentZCBF_a}\\
    \dot{\bs{x}}_j &= f_j(\bs{x}_j, \bs{x}_{\mathcal{N}_j}), \quad j \in V_f,\label{agentZCBF_b}
    \end{align}
\end{subequations}
where the functions $f_j(\cdot), j\in V$ and $ g_j(\cdot), j \in V_l$ are differentiable with respect to all their arguments, and their Jacobians are bounded on the domains $X \subset \mathbb{R}^{nM}$ and $U \subset \mathbb{R}^{mM}$.

In this setup, the leading agents’ movements are controlled by external inputs $\bs{u}_j, j  \in V_l$,  the remaining following agents simply follow their natural dynamics given in \eqref{agentZCBF_b}, maintaining coordination through their embedded controller. The goal is to navigate the agents towards the destination while maintaining connectivity among them in a distributed manner, meaning that each leader agent $j\in V_l$ computes its control input $\bs{u}_j$ based on the information available it. Each agent communicates with neighbors $k\in \mathcal{N}_j$ within a range $d_{max}$. Maintaining connectivity means preserving the  communication links, that is, ensuring that the distance between any two neighbor agents
$k,j\in V$, $ e_{kj}\in E$  remains within the communication range for all time
\begin{align}\label{eq:connectivity}
\| \bs{x}_k(t) - \bs{x}_j(t) \|_2 \le d_{max}, \quad \forall e_{kj} \in E,~ t \ge 0.
\end{align}

In this setup, leaders enforce connectivity through their control inputs. 
Consider a potential edge $e_{kj}\in E$ connecting a follower $k\in V_f$ and a leader $j\in V_l$. 
The distance dynamics on this edge are given by $\frac{d}{dt}\|\bs{x}_k-\bs{x}_j\|_2=\frac{(\bs x_k-\bs x_j)^\top(\dot{\bs x}_k-\dot{\bs x}_j)}{\|\bs x_k-\bs x_j\|_2}$ and depend on the states of the agents in $\mathcal{N}_k$ and $\mathcal{N}_j$ considering the \eqref{agentZCBF}.  Consequently, in order to enforce connectivity, leaders require two-hop neighborhood information, since the time derivative of the inter-agent distance on each edge depends not only on the states of the two incident agents but also on the states of their neighbors. This leads to the following assumption.
\begin{Assumption}[Leader information availability]\label{ass:neighborinfo}
    Leader agents $j\in V_l$, are assumed to have access to their two-hop neighborhood, i.e., their own state, the states of all agents $ k \in \mathcal{N}_j $, and the states of all agents in $ \bigcup_{k \in \mathcal{N}_j} \mathcal{N}_k $.  We define the corresponding \emph{local information index set} as
        $I_j^{loc} := \{ j \} \;\cup\; \mathcal{N}_j \;\cup\; \bigcup_{k \in \mathcal{N}_j} \mathcal{N}_k$,
    so that the vector of states available to agent $j$ are $\mathbf{x}_{I_j^{loc}} \in X_{I_j^{loc}}$.
\end{Assumption}
Based on the information available to them as specified in this assumption, each leader agent computes its control input to maintain connectivity with its neighbors.
To enforce connectivity, we associate each edge $e_{kj} \in E$, with a set of ZCBF functions
\begin{align*}
\{ h_r^{kj} : X_{I_{h_r^{kj}}} \to \mathbb{R} \}_{r=1}^{n_{kj}},
\end{align*}
whose corresponding safe sets will be used to maintain the connectivity constraint~\eqref{eq:connectivity}, where index set $I_{h_r^{kj}} \subseteq V$  denotes agents whose states appear in  $h_r^{kj}$, and $n_{kj} \in \mathbb{Z}^+$ is the total number of candidate CBFs associated with edge $e_{kj}$. We will show in subsequent analysis how to define the $n_{kj}$ for different edge class in Subsections~\ref{sec:lf}, \ref{sec:ll}, and \ref{sec:ff}.
Each candidate ZCBF $h_r^{kj}$ defines a safe set
\begin{align}\label{eq:safetyset}
S_r^{kj} = \{ \bs{x} \in X \mid h_r^{kj}(\bs{x}_{I_{h_r^{kj}}}) \ge 0 \},
\end{align}
representing the states in which the connectivity constraint between agents $k$ and $j$ is satisfied.  

In this setup, ensuring connectivity then reduces to designing the set of ZCBF functions $h_r^{kj}$ for each edge $e_{kj} \in E$ such that their corresponding safe sets $S_r^{kj}$ are forward invariant and within the intersection of safe sets $\bigcap_{r=1}^{n_{kj}} S_r^{kj},
$ the connectivity condition \eqref{eq:connectivity} is satisfied. In particular, if the initial relative positions of all connected agent pairs lie within their respective safe sets and these sets remain forward invariant, then the connectivity between neighboring agents is maintained for all future times.

In this setting, the dynamics of the agents are not explicitly known, so the functions $f_j$, $f_r$, and $g_j$ cannot be used directly. As a result, the standard ZCBF condition in \eqref{eq:ZCBFcond} cannot be applied in the usual way. Instead, implementation relies on previously collected input–state data for each candidate ZCBF to realize a data-driven version of the framework.
\begin{Assumption}[Data Availability]\label{ass_data}
    For each edge $e_{kj} \in E$ of the system, there is a non-empty finite set of continuously differentiable ZCBF candidates $\{h_r^{kj}: X_{I_{h_r^{kj}}} \rightarrow \mathbb{R}\}_{r=1}^{n_{kj}}$
    where $n_{kj}$ is the number of ZCBF candidates for edge $e_{kj}\in E$ and index set  $I_{h_r^{kj}} \subset V$ denotes agents whose states appear in  $h_r^{kj}$.
    
    Moreover, for each $h_r^{kj}$, a corresponding prior input–output dataset 
    \begin{align*}
    \mathcal{D}_r^{kj} = \{ (\dot{h}_r^{kj}(\bs{x}_{i,I_{\dot{h}_r^{kj}}}, \bs{u}_{i,I_{h_r^{kj}}}), \bs{x}_{i,I_{\dot{h}_r^{kj}}}, \bs{u}_{i,I_{h_r^{kj}}}) \}_{i=1}^N,\: r = 1,\dots,n_{kj}.
    \end{align*}
     is available with $\bs{x}_{i,I_{\dot{h}_r^{kj}}} \in X_{I_{\dot{h}_r^{kj}}}$, and $ \bs{u}_{i,I_{h_r^{kj}}} \in U_{I_{h_r^{kj}}}$. Here, $I_{\dot{h}_r^{kj}}$ denotes the index set of agents whose states affect the time derivative $\dot{h}_r^{kj}$. For brevity each 
    $\dot{h}_r^{kj}(\bs{x}_{i,I_{\dot{h}_r^{kj}}}, \bs{u}_{i,I_{h_r^{kj}}})$ will be denoted by $\dot{h}^{kj}_{r,i}$. 
\end{Assumption}
Note that the control inputs affecting $\dot{h}_r^{kj}$ can only be those of the agents appearing in $h_r^{kj}$, while the states influencing $\dot{h}_r^{kj}$ may include additional agents due to couplings in the dynamics \eqref{agentZCBF}; accordingly, in the notation $\dot{h}_r^{kj}(\bs{x}_{i,I_{\dot{h}_r^{kj}}}, \bs{u}_{i,I_{h_r^{kj}}})$, the input index set is taken from $h_r^{kj}$, and a separate index set $I_{\dot{h}_r^{kj}}$ is introduced for states.

Rather than trying to learn the full $n$-dimensional agent dynamics $\dot{\bs{x}}_j$ for $j = 1,\dots,M$, this work focuses on learning the one-dimensional functions $\dot{h}_r^{kj}$ directly from the dataset, under Assumption~\ref{ass_data}. Restricting attention to $\dot{h}_r^{kj}$ lowers the dimensionality of the learning task, which is especially helpful in multi-agent systems with tightly coupled, high-dimensional dynamics. Because the true mappings $\dot{h}_r^{kj}(\bs{x}_{i,I_{\dot{h}_r^{kj}}}, \bs{u}_{i,I_{h_r^{kj}}})$ are unknown, the standard ZCBF condition in \eqref{eq:ZCBFcond} cannot be applied directly. To address this, an over-approximation of these functions is learned from the available data, with explicit quantification of the generalization error. Following \cite{jin2023robust}, a non-parametric learning scheme is used, supported by a continuity assumption.
\begin{Assumption}\label{lipschitz_cbf}
    For all agents $e_{kj}\in E$ and $r=1\cdots n_{kj}$, the functions $\dot{h}_r^{kj} : X_{I_{\dot{h}_r^{kj}}} \times U_{I_{h_r^{kj}}} \rightarrow \mathbb{R}$  are  differentiable with respect to $\bs{x}_{I_{\dot{h}_r^{kj}}}$ and $\bs{u}_{I_{h_r^{kj}}}$, and their Jacobians  are globally bounded with known finite-valued matrices $\underline{J}^{kj}_{r}, \overline{J}^{kj}_{r} \in \mathbb{R}^{(n+m)|I_{\dot{h}_r^{kj}}|}$, i.e.,
    \begin{align*}
    \underline{J}^{kj}_{r} \le \nabla\dot{h}_r^{kj}(\bs{x}_{I_{\dot{h}_r^{kj}}},\bs{u}_{I_{h_r^{kj}}}) \le \overline{J}^{kj}_{r}, 
    \end{align*}
    for all $\bs{x} \in X \subset \mathbb{R}^{nM}$ and $\bs{u} \in U \subset \mathbb{R}^{mM}$, where the inequalities are interpreted componentwise.
\end{Assumption}
Note that a sufficient condition for Assumption \ref{lipschitz_cbf} to hold for $\dot{h}_r^{kj}$ is that they also hold for functions $h_r^{kj}$, $f_j$, and $g_j$. We can now state the problem that is addressed in this paper formally. 

\textbf{Problem:} Consider an unknown multi-agent system of the form \eqref{agentZCBF}. The goal of the \emph{distributed connectivity maintenance} problem is to ensure that the distance between any two neighboring agents $e_{kj} \in E$ always remains below the maximum communication distance $d_{\max}$.  The task is to design candidate CBFs $\{ h_r^{kj} : X_{I_{h_r^{kj}}} \to \mathbb{R} \}_{r=1}^{n_{kj}}$ satisfying Assumption~\ref{lipschitz_cbf} for each edge $e_{kj}\in E$, such that within their associated safety sets $\{ S_r^{kj} : X_{I_{h_r^{kj}}} \to \mathbb{R} \}_{r=1}^{n_{kj}}$ the connectivity condition is enforced, and the invariance of these sets can be guaranteed in a distributed data-driven manner. Specifically, each leader $j\in V_l$ computes its control input using the available datasets satisfying Assumption~\ref{ass_data} and the two-hop neighborhood information available to it as specified in Assumption~\ref{ass:neighborinfo} without knowledge of the control inputs of other leaders and the combined actions of all leaders ensure that the relevant safety sets remain forward invariant. 
\section{Main Results}
The first step is to introduce the 3D-ZCBF framework for the unknown multi-agent system in \eqref{agentZCBF}. This framework allows CBFs to be defined for various types of edges and provides a way to formulate the associated safety conditions.
\begin{definition}[3D-ZCBF] \label{def:cbfcentral} 
    For an unknown control-affine multi-agent system~\eqref{agentZCBF}, a continuously differentiable function $h_r^{kj}: X_{I_{h_r^{kj}}} \rightarrow \mathbb{R}$ whose time derivative $\dot{h}_r^{kj}$ satisfies Assumption~\ref{lipschitz_cbf} with an input-output dataset $\mathcal{D}^{kj}_r$ satisfying Assumption~\ref{ass_data},  is called a \emph{bounded Jacobian robust distributed data-driven control barrier function (3D-ZCBF)} for the safety set $S_r^{kj}$ defined in \eqref{eq:safetyset}, if the followings hold:
    \begin{enumerate}
        \item \textbf{Data-driven condition:} There exists a class $\mathcal{K}_{\infty}$ function $\alpha(\cdot)$ and a data index $i^{\bs{x}}_{kj}\in\mathbb{Z}^+_N$ such that
        \begin{align}\label{eq:cbfdd_cond}
        &\sup_{\bs{u}\in U}  \dot{h}_{r,i^{\bs{x}}_{kj}}^{kj} + \underline{J}^{kj}_{r,\bs{x}}  \Delta_{\bs{x}_{i^{\bs{x}}_{kj},I_{\dot{h}_r^{kj}}}}^+ - \overline{J}^{kj}_{r,\bs{x}} \Delta_ {\bs{x}_{i^{\bs{x}}_{kj},I_{\dot{h}_r^{kj}}}}^- 
        + \underline{J}^{kj}_{r,\bs{x}} \Delta_ {\bs{u}_{i^{\bs{x}}_{kj},I_{h_r^{kj}}}}^+\nonumber\\ 
        &- \overline{J}^{kj}_{r,\bs{u}} \Delta_ {\bs{u}_{i^{\bs{x}}_{kj},I_{h_r^{kj}}}}^-
        \ge -\alpha\big(h(\bs{x}_{I_{h_r^{kj}}})\big),
    \end{align}
     for all $\bs{x} \in X$ where $\Delta_{\bs{x}_{i^{\bs{x}}_{kj},I_{\dot{h}_r^{kj}}}} := \bs{x}_{I_{\dot{h}_r^{kj}}} - \bs{x}_{i^{\bs{x}}_{kj},I_{\dot{h}_r^{kj}}}$ and $\Delta_ {\bs{u}_{i^{\bs{x}}_{kj},I_{h^{kj}_{1}}}} := \bs{u}_{I_{h_r^{kj}}} - \bs{u}_{i^{\bs{x}}_{kj},I_{h_r^{kj}}}$.
      \item \textbf{Distributed input-decoupled condition:} There exist functions
    \begin{align*}
        \phi_{\ell,r}^{kj}: X_{I_{\ell}^{\rm loc}} \times U_{\ell} \times \mathbb{Z}^+_N \to \mathbb{R}, \; \psi_{\ell,r}^{kj}: X_{I_{\ell}^{\rm loc}} \to \mathbb{R}, \; {\ell} \in I_{h_r^{kj}}\cap V_l,
    \end{align*}
    such that the original condition \eqref{eq:cbfdd_cond} can be equivalently expressed as a sum of local terms:
     \begin{align}
        & \sum_{{\ell} \in I_{h_r^{kj}}\cap V_l} \phi_{\ell,r}^{kj}(\mathbf{x}_{I_{\ell}^{\rm loc}}, u_{\ell},i^{\bs{x}}_{kj})=  \dot{h}_{r,i^{\bs{x}}_{kj}}^{kj} + \underline{J}^{kj}_{r,\bs{x}}  \Delta_{\bs{x}_{i^{\bs{x}}_{kj},I_{\dot{h}_r^{kj}}}}^+ \nonumber \\
        &- \overline{J}^{kj}_{r,\bs{x}} \Delta_ {\bs{x}_{i^{\bs{x}}_{kj},I_{\dot{h}_r^{kj}}}}^- 
        + \underline{J}^{kj}_{r,\bs{x}} \Delta_ {\bs{u}_{i^{\bs{x}}_{kj},I_{h_r^{kj}}}}^+ 
        - \overline{J}^{kj}_{r,\bs{u}} \Delta_ {\bs{u}_{i^{\bs{x}}_{kj},I_{h_r^{kj}}}}^-,\nonumber \\
        &\sum_{{\ell} \in I_{h_r^{kj}}\cap V_l} \psi_{\ell,r}^{kj}(\mathbf{x}_{I_{\ell}^{\rm loc}})=\alpha\big(h(\bs{x}_{I_{h_r^{kj}}})\big),
    \end{align}
    Consequently, the 3D-ZCBF can be enforced in a distributed manner by each leader agent independently:
    \begin{align}\label{eq:decoupledZCBF}
        \sup_{u_{\ell} \in U_{\ell}} \phi_{\ell,r}^{kj}(\mathbf{x}_{I_{\ell}^{\rm loc}}, u_{\ell},i^{\bs{x}}_{kj}) \ge -\psi_{\ell,r}^{kj}(\mathbf{x}_{I_{\ell}^{\rm loc}}), \quad \forall {\ell} \in I_{h_r^{kj}}\cap V_l.
    \end{align}
    \end{enumerate}
The corresponding safe-input set is
    \begin{align*}
       & K_{S_r^{kj},i^{\bs{x}}_{kj}}(\bs{x}) := U_{S_r^{kj},^*_{kj}}(\bs{x}), \nonumber \\
        &U_{S_r^{kj},i^{\bs{x}}_{kj}}(\bs{x}) := \{ \bs{u} \in U \;|\; 
        \dot{h}_{r,i^{\bs{x}}_{kj}}^{kj} + \underline{J}^{kj}_{r,\bs{x}}  \Delta_{\bs{x}_{i^{\bs{x}}_{kj},I_{\dot{h}_r^{kj}}}}^+ - \overline{J}^{kj}_{e,\bs{x}} \Delta_ {\bs{x}_{i^{\bs{x}}_{kj},I_{\dot{h}_r^{kj}}}}^- \\
        &+ \underline{J}^{kj}_{r,\bs{u}} \Delta_ {\bs{u}_{i^{\bs{x}}_{kj},I_{h_r^{kj}}}}^+ 
        - \overline{J}^{kj}_{r,\bs{u}} \Delta_  {\bs{u}_{i^{\bs{x}}_{kj},I_{h_r^{kj}}}}^- 
        \ge -\alpha\big(h(\bs{x}_{I_{h_r^{kj}}})\big) \}.
    \end{align*}
\end{definition}
The data-driven ZCBF condition \eqref{eq:cbfdd_cond} applies the framework in \cite{jin2023robust} to multi-agent systems with unknown dynamics, relying entirely on the data at hand to guarantee safety. The distributed form, given in \eqref{eq:decoupledZCBF}, lets each leader enforce safety with information that is locally accessible, so no central coordination is needed.

Note that \cite{jin2023robust} formulates the data-driven condition for CBFs as a maximum over the entire dataset, rather than using a single data index. In practice, this maximization is implemented by solving a separate optimization for each data point to obtain the corresponding optimal input, and then selecting among these candidates that achieves the best value of a cost function (Equation 12 in \cite{jin2023robust}). Since this is already computationally intensive for a single ZCBF \cite{jin2023robust}, it becomes prohibitive when multiple CBFs are involved, at least one for each edge. Moreover, in a leader-follower multi-agent setup, where each leader computes its input independently without knowledge of other leaders’ inputs, dataset-maximization is not directly applicable, as it requires a joint optimization over all agents’ inputs. To overcome both the computational and distributed implementation challenges, we adopt a single data index definition. However, this condition restricts the feasible input set more than a dataset-maximization approach, since the left-hand side of \eqref{eq:cbfdd_cond} is evaluated at a single data point rather than the maximal value over the entire dataset. Now, we present the theorem establishing that any ZCBF candidate $h_r^{kj}$ satisfying \eqref{eq:cbfdd_cond} renders its safety set $S_r^{kj}$ robustly control invariant for the system~\eqref{agentZCBF}.
\begin{Theorem}\label{th:centralcbf}
    For the unknown control-affine multi-agent system~\eqref{agentZCBF}  with a 3D-ZCBF $h_r^{kj}: X_{I_{h_r^{kj}}} \rightarrow \mathbb{R}$, as defined in Definition~\ref{def:cbfcentral}, and its associated safety set $S_r^{kj}$, any controller $\bs{u} \in K_{S_r^{kj},i^{\bs{x}}_{kj}}(\bs{x})$, where $i^{\bs{x}}_{kj}$ is an index satisfying \eqref{eq:cbfdd_cond}, renders the set $S_r^{kj}$ \emph{robustly controlled invariant}.
\end{Theorem}
\begin{proof}
    For any  $(\bs{x},\bs{u}) \in X \times U$, consider the 3D-ZCBF  time derivative
    $f(\bs{x},\bs{u}) := \dot{h}_r^{kj}(\bs{x}_{I_{\dot{h}_r^{kj}}},\bs{u}_{I_{h_r^{kj}}})$.
    By Assumption~\ref{lipschitz_cbf}, $f(\cdot)$ is continuously differentiable
    with bounded Jacobian matrices with respect to $\bs{x}$ and $\bs{u}$.
    
    For any data index $i\in \mathbb{Z}^+_N$ and any admissible input $\bs{u} \in U$,
    the mean value theorem implies
    \begin{equation*}
        f(\bs{x},\bs{u}) = f(\bs{x}_{i},\bs{u}_{i})
        + \nabla_{\bs{x}} f(\xi_x)(\bs{x}-\bs{x}_{i})
        + \nabla_{\bs{u}} f(\xi_u)(\bs{u}-\bs{u}_{i})
    \end{equation*}
    for some $\xi_x,\xi_u$ on the line segments connecting $(\bs{x},\bs{x}_{i})$ and $(\bs{u},\bs{u}_{i})$ respectively. Using the componentwise bounds from Assumption~\ref{lipschitz_cbf} and applying Proposition~2 in~\cite{jin2022data}, which provides componentwise lower and upper bounds for a matrix-vector product given interval bounds on the matrix and the vector, each term is bounded as
    \begin{align*}
        &\nabla_{\bs{x}} f(\xi_x)(\bs{x}-\bs{x}_{i}) 
        \ge \underline{J}^{kj}_{r,\bs{x}} \Delta_{\bs{x}_{i,I_{\dot{h}_r^{kj}}}}^+ 
        - \overline{J}^{kj}_{r,\bs{x}} \Delta_{\bs{x}_{i,I_{\dot{h}_r^{kj}}}}^-,
        \\
        &\nabla_{\bs{u}} f(\xi_u)(\bs{u}-\bs{u}_{i}) 
        \ge \underline{J}^{kj}_{r,\bs{u}} \Delta_{\bs{u}_{i,I_{h_r^{kj}}}}^+ 
        - \overline{J}^{kj}_{r,\bs{u}} \Delta_{\bs{u}_{i,I_{h_r^{kj}}}}^-.
    \end{align*}
    Combining these gives a lower bound on $f(\bs{x},\bs{u})$ as
    \begin{align*}
    f(\bs{x},\bs{u}) \ge\;
    \dot{h}_{r,i}^{kj}
    &+ \underline{J}^{kj}_{r,\bs{x}} \Delta_{\bs{x}_{i,I_{\dot{h}_r^{kj}}}}^+
    - \overline{J}^{kj}_{r,\bs{x}} \Delta_{\bs{x}_{i,I_{\dot{h}_r^{kj}}}}^- \nonumber\\
    &+ \underline{J}^{kj}_{r,\bs{u}} \Delta_{\bs{u}_{i,I_{h_r^{kj}}}}^+
    - \overline{J}^{kj}_{r,\bs{u}} \Delta_{\bs{u}_{i,I_{h_r^{kj}}}}^- .
    \end{align*}
    Let $i^{\bs{x}}_{kj}$ denote a data index satisfying the data-driven condition~\eqref{eq:cbfdd_cond}.  
    Replacing the generic index $i$ in the lower bound with $i^{\bs{x}}_{kj}$ and taking the supremum over $\bs{u} \in U$ we get
   \begin{align}\label{eq:proof}
    &\sup_{u \in U}(\dot{h}_r^{kj}(\bs{x}_{I_{\dot{h}_r^{kj}}}),\bs{u}_{I_{h_r^{kj}}}) \ge \sup_{\bs{u}\in U}  \dot{h}_{r,i^{\bs{x}}_{kj}}^{kj} + \underline{J}^{kj}_{r,\bs{x}}  \Delta_{\bs{x}_{i^{\bs{x}}_{kj},I_{\dot{h}_r^{kj}}}}^+ - \overline{J}^{kj}_{r,\bs{x}} \Delta_ {\bs{x}_{i^{\bs{x}}_{kj},I_{\dot{h}_r^{kj}}}}^-\\ \nonumber
        &+ \underline{J}^{kj}_{r,\bs{x}} \Delta_ {\bs{u}_{i^{\bs{x}}_{kj},I_{h_r^{kj}}}}^+ 
        - \overline{J}^{kj}_{r,\bs{u}} \Delta_ {\bs{u}_{i^{\bs{x}}_{kj},I_{h_r^{kj}}}}^-.
    \end{align}
We can then apply the data-driven condition~\eqref{eq:cbfdd_cond} to the right-hand side, yielding
    \begin{equation*}
        \sup_{\bs u \in U} (\dot{h}_r^{kj}(\bs{x}_{I_{\dot{h}_r^{kj}}},\bs{u}_{I_{h_r^{kj}}})) \ge -\alpha\big(h(\bs{x}_{I_{h_r^{kj}}})\big).
    \end{equation*}
    Then, it follows that standard ZCBF inequality~\eqref{eq:ZCBFcond} holds, which implies that
    the safety set $S_r^{kj}$ is robustly controlled invariant.
\end{proof}
Building on the general CBF formulation established above, the ZCBF functions must be designed separately for each edge type, because the time derivative of the connectivity constraint \eqref{eq:connectivity} involves different combinations of states and control inputs depending on whether the connected agents are leaders or followers. Moreover, the original data-driven inequality \eqref{eq:cbfdd_cond} may involve coupled control inputs from multiple agents, which conflicts with the distributed implementation requirement. Therefore, these conditions needs to be reformulated in decoupled form \eqref{eq:decoupledZCBF} when necessary. In the following, we define the ZCBF functions $h_r^{kj}$ for each edge type so that their time derivatives $\dot{h}_r^{kj}$ involve the necessary control inputs and have bounded Jacobians, thereby ensuring that the 3D-ZCBF conditions \eqref{eq:cbfdd_cond} can be enforced for the forward invariance of their safety sets. Within their corresponding safety sets, these ZCBFs are designed to maintain the desired connectivity between agents. Then, we derive their associated decoupled conditions suitable for distributed implementation. The distributed conditions for each edge are subsequently collected and presented together in Theorem \ref{prop:robust-invariance}, providing a unified statement of the constraints that each leader must satisfy across all its relevant edges.
\subsection{Leader-follower edges}\label{sec:lf}
For each leader-follower edge $e_{kj} \in E_{lf}$, where $k \in V_f$ and $j \in V_l$, a single ZCBF candidate is considered, i.e., $n_{kj} = 1$. The corresponding ZCBF candidate is defined as 
\begin{align}\label{eq:cbflf}
    & h^{kj}_{1}(\bs{x}_{I_{h^{kj}_{1}}}) := d_{\max}^2 - \|\bs{\bar{x}}_{kj}\|_2^2, 
    \quad e_{kj} \in E_{lf},
\end{align}
where $\bs{\bar x}_{kj}:=\bs{x}_k-\bs{x}_j$, $I_{h^{kj}_{1}}=\{k, j\}$. {Note that, since the state and input sets $X,U$ are bounded,  the time derivative of ZCBF candidate $h_1^{kj}$ satisfies the bounded Jacobian requirement given in  Assumption~\ref{lipschitz_cbf}  which is needed for enforcing 3D-ZCBF condition~\eqref{eq:cbfdd_cond}. This follows from the boundedness of the state differences $\bs{\bar x}_{kj}$ in the bounded domain $X$ and the bounded Jacobians of the dynamics.} For this function, it follows from \eqref{agentZCBF} that the index set of its time derivative is $
I_{\dot{h}^{kj}_{1}} = \{k,j\} \cup \mathcal{N}_k \cup \mathcal{N}_j
$. Among these agents, the only nonzero input term appearing in $\dot{h}^{kj}_{1}$ is $\bs{u}_j$ since agent $k$ is a follower and the dynamics of $k$ and $j$ depend only on the states of $\mathcal{N}_k$ and $\mathcal{N}_j$, but not on their inputs.  Consequently, in the distributed formulation of the 3D-ZCBF condition \eqref{eq:decoupledZCBF}, the only leader enforcing the inequality is $j$, i.e., 
$\ell \in I_{h_1^{kj}} \cap V_l = \{j\}.$
Therefore, the summation over leader contributions in the general distributed definition reduces to a single term, and the condition is already naturally decoupled, requiring no further separation.
 The data index $i^{\bs{x}}_{kj}$ can  be chosen as
\begin{align}\label{eq:istarlf}
        i^{\bs{x}}_{kj} = \arg\max_{i \in \mathbb{Z}_N^+} \;&
        \dot{h}^{kj}_{1,i} + \underline{J}^{kj}_{1,\bs{x}} \Delta_{\bs{x}_{i,I_{\dot{h}^{kj}_{1}}}}^+ - \overline{J}^{kj}_{1,\bs{x}} \Delta_{\bs{x}_{i,I_{\dot{h}^{kj}_{1}}}}^-,  
\end{align}
  where $\Delta_{\bs{x}_{i,I_{\dot{h}^{kj}_{1}}}} := \bs{x}_{I_{\dot{h}^{kj}_{1}}} - \bs{x}_{i,I_{\dot{h}^{kj}_{1}}}$. This choice picks the data sample that maximizes the state-dependent part of the lower bound in \eqref{eq:proof}, identifying the point that provides the tightest guaranteed lower bound on $\sup_{\bs{u}\in U}\dot{h}^{kj}_1$ using only state measurements. Since this selection does not depend on any agent’s input, it yields a condition suitable for distributed implementation. 
\subsection{Leader-leader edges}\label{sec:ll}
For each leader-leader edge $e_{kj} \in E_{lf}$, where $k ,j \in V_l$, a single ZCBF candidate is considered, i.e., $n_{kj} = 1$. We consider the same ZCBF candidate as in leader-follower case as
\begin{align}\label{eq:cbfll}
    & h^{kj}_{1}(\bs{x}_{I_{h^{kj}_{1}}}) := d_{\max}^2 - \|\bs{\bar{x}}_{kj}\|_2^2, 
    \quad e_{kj} \in E_{ll},
\end{align}
where $\bs{\bar x}_{kj}:=\bs{x}_k-\bs{x}_j$, $I_{h^{kj}_{1}}=\{k, j\}$. The index set of its time derivative is $
I_{\dot{h}^{kj}_{1}} = \{k,j\} \cup \mathcal{N}_k \cup \mathcal{N}_j
$. Since both nodes $k$ and $j$ are leaders, the derivative $\dot{h}^{kj}_{1}$ depends on the two inputs $\bs{u}_k$ and $\bs{u}_j$. Consequently, the data-driven ZCBF condition \eqref{eq:cbfdd_cond} depends on both inputs. To make it suitable for distributed implementation, the condition is divided between agents $k$ and $j$, so that each leader enforces its own part independently. Adding these contributions together recovers the full condition. These decoupled constraints are formalized in the following lemma.
\begin{Lemma}\label{lemma:llcbf}
     Let $e_{kj} \in E_{ll}$ where $k,j \in V_l$ be a leader–leader edge. For an unknown control-affine multi-agent system~\eqref{agentZCBF}, consider a continuously differentiable function $h^{kj}_{1}: X_{I_{h^{kj}_{1}}} \rightarrow \mathbb{R}$ as defined in \eqref{eq:cbfll} with an input-output dataset $\mathcal{D}^{kj}_1$, satisfying Assumption~\ref{ass_data}. Let $i^{\bs{x}}_{kj} \in \mathbb{Z}_N^+$ be the index that maximizes 
    \begin{align}\label{eq:istarll}
        i^{\bs{x}}_{kj} = \arg\max_{i \in \mathbb{Z}_N^+} \;&
        \dot{h}^{kj}_{1,i} + \underline{J}^{kj}_{1,\bs{x}} \Delta_{\bs{x}_{i,I_{\dot{h}^{kj}_{1}}}}^+ - \overline{J}^{kj}_{1,\bs{x}} \Delta_{\bs{x}_{i,I_{\dot{h}^{kj}_{1}}}}^-,  
    \end{align}
    where $\Delta_{\bs{x}_{i,I_{\dot{h}^{kj}_{1}}}} := \bs{x}_{I_{\dot{h}^{kj}_{1}}} - \bs{x}_{i,I_{\dot{h}^{kj}_{1}}}$. Then $h^{kj}_{1}$ is a 3D-ZCBF for edge $e_{kj}$ 
    if there exist  a class $\mathcal{K}_{\infty}$ function $\alpha(\cdot)$ and constants $\beta_k\ge 0, \beta_j \ge 0, \beta_k + \beta_j = 1$ such that 
\begin{subequations}\label{eq:cbfcondll}
    \begin{align}
        &\sup_{\bs{u}_k \in U_k} \beta_k [\dot{h}^{kj}_{1,i^{\bs{x}}_{kj}} +\underline{J}^{kj}_{1,\bs{x}_{\mathcal{N}_{k  j}}} \Delta^+_{\bs{x}_{i^{\bs{x}}_{kj},\mathcal{N}_{k  j}}} - \overline{J}^{kj}_{1,\bs{x}_{\mathcal{N}_{k  j}}} \Delta^-_{\bs{x}_{i^{\bs{x}}_{kj},\mathcal{N}_{k  j}}}] \nonumber \\
        &+ \underline{J}^{kj}_{1,\bs{x}_k} \Delta^+_{\bs{x}_{i^{\bs{x}}_{kj},k}} - \overline{J}^{kj}_{1,\bs{x}_k} \Delta^-_{\bs{x}_{i^{\bs{x}}_{kj},k}}  
        + \underline{J}^{kj}_{1,\bs{u}_k} \Delta^+_{\bs{u}_{i^{\bs{x}}_{kj},k}}
        - \overline{J}^{kj}_{1,\bs{u}_k} \Delta^-_{\bs{u}_{i^{\bs{x}}_{kj},k}}  \nonumber\\
        &+ \underline{J}^{kj}_{1,\bs{x}_{\mathcal{N}_{k \setminus j}}} \Delta^+_{\bs{x}_{i^{\bs{x}}_{kj},\mathcal{N}_{k \setminus j}}} - \overline{J}^{kj}_{1,\bs{x}_{\mathcal{N}_{k \setminus j}}} \Delta^-_{\bs{x}_{i^{\bs{x}}_{kj},\mathcal{N}_{k \setminus j}}} \nonumber \\
        &\ge - \beta_k \, \alpha(h^{kj}_{1}(\bs{x}_{I_{h_r^{kj}}})), \forall\bs{x} \in X,\\
       &\sup_{\bs{u}_j \in U_j} \beta_j [\dot{h}^{kj}_{1,i^{\bs{x}}_{kj}} +\underline{J}^{kj}_{1,\bs{x}_{\mathcal{N}_{k  j}}} \Delta^+_{\bs{x}_{i^{\bs{x}}_{kj},\mathcal{N}_{k  j}}} - \overline{J}^{kj}_{1,\bs{x}_{\mathcal{N}_{k  j}}} \Delta^-_{\bs{x}_{i^{\bs{x}}_{kj},\mathcal{N}_{k  j}}} ] \nonumber \\
        &+ \underline{J}^{kj}_{1,\bs{x}_j} \Delta^+_{\bs{x}_{i^{\bs{x}}_{kj},j}} - \overline{J}^{kj}_{1,\bs{x}_j} \Delta^-_{\bs{x}_{i^{\bs{x}}_{kj},j}}  
        + \underline{J}^{kj}_{1,\bs{u}_j} \Delta^+_{\bs{u}_{i^{\bs{x}}_{kj},j}}
        - \overline{J}^{kj}_{\bs{u}_j} \Delta^-_{\bs{u}_{i^{\bs{x}}_{kj},j}}  \nonumber\\
        &+ \underline{J}^{kj}_{1,\bs{x}_{\mathcal{N}_{j \setminus k}}} \Delta^+_{\bs{x}_{i^{\bs{x}}_{kj},\mathcal{N}_{j \setminus k}}} - \overline{J}^{kj}_{1,\bs{x}_{\mathcal{N}_{j \setminus k}}} \Delta^-_{\bs{x}_{i^{\bs{x}}_{kj},\mathcal{N}_{j \setminus k}}} \nonumber \\
        &\ge - \beta_j \, \alpha(h^{kj}_{1}(\bs{x}_{I_{h_r^{kj}}})), \forall\bs{x} \in X,
    \end{align}
    \end{subequations}
    where the left-hand sides and right-hand sides define the functions $\phi_k, \phi_j: X \to \mathbb{R}$,
    where $\Delta_{\bs{x}_{i^{\bs{x}}_{kj},A}} := \bs{x}_A - \bs{x}_{i^{\bs{x}}_{kj},A}$, $
    \Delta_{\bs{u}_{i^{\bs{x}}_{kj},A}} := \bs{u}_A - \bs{u}_{i^{\bs{x}}_{kj},A}$ for any agent index set $A \subseteq \{k,j\} \cup \mathcal{N}_k \cup \mathcal{N}_j$.
\end{Lemma}
The inequalities in \eqref{eq:cbfcondll} define the distributed functions 
$\phi_k, \phi_j$ and $\psi_k, \psi_j$, 
such that \eqref{eq:cbfdd_cond} is equivalently expressed in the decoupled form \eqref{eq:decoupledZCBF}. The weights $\beta_k, \beta_j$ distribute the enforcement of the safety constraint between the two leaders so that each satisfies only a portion of the 3D-ZCBF condition \eqref{eq:cbfdd_cond}. Satisfaction of the distributed conditions for both agents guarantees that the overall 3D-ZCBF condition holds, as the sum of the two inequalities equals the full condition. 

\subsection{Follower-follower edges}\label{sec:ff}
Edges between followers need a different approach because neither agent has an external input in its dynamics. Using the same ZCBF candidate as before \eqref{eq:cbfll} would produce a time derivative without any control terms. While the first derivative could still be evaluated for safety, it would not constrain the available inputs, making it insufficient to enforce safety. Ensuring that the safety set remains invariant would require conditions on the second derivative, where input terms appear. This, however, demands data on second derivatives, which is often impractical. Since derivatives are typically estimated from time-series data, using the second derivative would likely introduce considerable noise. To avoid this, alternative ZCBF candidates are designed so that input terms already appear in the first derivative, requiring only first-derivative data. To formalize this, we first make the following assumption.
\begin{Assumption}\label{as:ffleaders}
    For each follower–follower edge $e_{kj} \in E_{ff}$, it is assumed that followers $k$ and $j$ each have at least one distinct leader in their neighborhood, i.e.,$ (\mathcal{N}_k \cap V_l)\setminus \mathcal{N}_j \neq \emptyset$ and $(\mathcal{N}_j \cap V_l)\setminus \mathcal{N}_k \neq \emptyset$, with at least one leader for each follower.   
    Furthermore, for the purpose of defining ZCBF candidates, one leader from each follower’s neighborhood is selected, denoted $(\ell_k, \ell_j)$ with $\ell_k \in \mathcal{N}_k \cap V_l$ and $\ell_j \in \mathcal{N}_j \cap V_l$, with $\ell_k \neq \ell_j$. 
\end{Assumption}

Given this fixed leader pair, the relative distance between the followers $\bs{\bar x}_{kj} := \bs{x}_k - \bs{x}_j$ is split into two components: one parallel to the vector $\bs{\bar x}_{\ell_k\ell_j} := \bs{x}_{\ell_k} - \bs{x}_{k}$ connecting the leaders $\ell_k, \ell_j$ and one orthogonal to it as
\begin{align}
    &\bs{\bar x}_{kj} = \bs{\bar x}_{kj}^{\parallel} + \bs{\bar x}_{kj}^{\perp}, 
    \quad
    \bs{\bar x}_{kj}^{\parallel} := (\bs{\bar x}_{kj}^\top  \bs{\bar x}^\wedge_{\ell_k\ell_j})\, \bs{\bar x}^\wedge_{\ell_k\ell_j}, 
    \quad
    \bs{\bar x}_{kj}^{\perp} := \bs{\bar x}_{kj} - \bs{\bar x}_{kj}^{\parallel},
\end{align}
{where $\bs{\bar x}^\wedge_{\ell_k\ell_j} := \frac{\bs{\bar x}_{\ell_k\ell_j}}{\|\bs{\bar x}_{\ell_k\ell_j}\|_2}$ 
is well-defined with $\|\bs{x}_{\ell_k} - \bs{x}_{\ell_j}\|_2 >  \epsilon$ for some small $\epsilon > 0$.
}
For follower-follower edges, we consider two ZCBF candidates, i.e., $n_{kj} = 2$, corresponding to first-order barrier functions along these parallel and orthogonal components, defined as
\begin{align}
     h^{kj}_{1}(\bs{x}_{I_{h_1^{kj}}}) &:= \frac{d_{max}^2}{2} - \|\bs{\bar x}_{kj}^{\parallel}\|_2^2, \quad e_{kj} \in E_{ll}, \label{eq:cbfffparalel}\\
     h^{kj}_{2}(\bs{x}_{I_{h_2^{kj}}}) &:= \frac{d_{max}^2}{2} - \|\bs{\bar x}_{kj}^{\perp}\|_2^2, \quad e_{kj} \in E_{ll}, \label{eq:cbfffperp}
\end{align}
where $I_{h^{kj}_{1}} =I_{h^{kj}_{2}} = \{k,j,\ell_k,\ell_j\}$. This geometric decomposition is illustrated in Fig~\ref{fig:ff_zcbf_geometry}. For these functions, it follows from \eqref{agentZCBF} that the index set of its time derivative is $
I_{\dot{h}^{kj}_{1}} =I_{\dot{h}^{kj}_{2}} = \{k,j,\ell_k,\ell_j\} \cup \mathcal{N}_k \cup \mathcal{N}_j \cup \mathcal{N}_{\ell_k} \cup \mathcal{N}_{\ell_j}$. {Since state and input sets $X,U$ are bounded,  the ZCBF candidates $h_1^{kj}, h_2^{kj}$ satisfies the bounded Jacobian requirement of Assumption~\ref{lipschitz_cbf} for its time derivative $\dot{h}^{kj}_1$. This follows from the boundedness of state differences in $X$ and the 
fact that the derivative of $\bs{\bar x}^\wedge_{\ell_k\ell_j}$ with respect to 
$\bs{x}_{\ell_k},\bs{x}_{\ell_j}$ remains bounded.
} 
It can be seen that inside the intersection of their safety sets $S_1^{kj},S_2^{kj}$, the connectivity condition \eqref{eq:connectivity}
is satisfied, thus maintaining connectivity between followers. For these ZCBF candidates, the first time derivatives $\dot{h}^{kj}_{1}, \dot{h}^{kj}_{2}$ include input terms, since the barrier functions depend on the states of the leader agents. This decomposition ensures that each barrier function is of first order, allowing direct use in a data-driven ZCBF framework where only the first derivative of the barrier function is available as in Assumption \ref{ass_data}.  {As ZCBF candidates $h_1^{kj}$ and $h_2^{kj}$ are continuously differentiable for all $\bs{x}_{I_{h_1^{kj}}} \in X_{I_{h_1^{kj}}}$, $\bs{x}_{I_{h_2^{kj}}}\in X_{I_{h_1^{kj}}}$, 
the ZCBF conditions \eqref{eq:ZCBFcond} are well-defined, so that the associated forward invariance conditions can be enforced through our data-driven framework.}

\begin{figure}[t]
    \centering
    \input{figures/ff.tikz}
    \caption{Geometric decomposition used in the follower--follower ZCBF construction.
    The relative displacement $\bs{\bar x}_{kj}$ is decomposed into components parallel
    and orthogonal to the leader--leader direction $\bs{\bar x}_{\ell_k\ell_j}$, yielding the
barrier functions $h^{kj}_1$ and $h^{kj}_2$ in
\eqref{eq:cbfffparalel}–\eqref{eq:cbfffperp}.}
    \label{fig:ff_zcbf_geometry}
\end{figure}
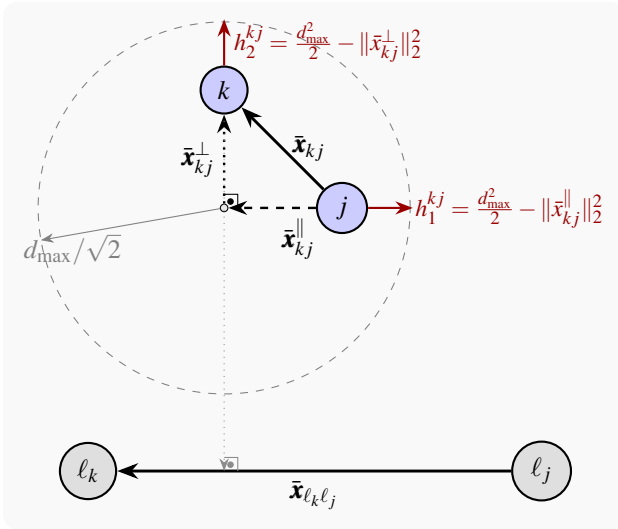


For these functions, the 3D-ZCBF condition \eqref{eq:cbfdd_cond} links the inputs of both leader agents $\bs{u}_{\ell_k}$ and $\bs{u}_{\ell_j}$, creating coupled inequalities. Following the approach used for leader–leader edges,  this condition is split between the two agents to obtain decoupled constraints \eqref{eq:decoupledZCBF}. Before formalizing the follower–follower ZCBF conditions, we specify how the data sample used for bounding $\dot{h}^{kj}_{1}$ is chosen. 
For each follower–follower edge $e_{kj} \in E_{ff}$, an index $i^{\bs{x}}_{kj} \in \mathbb{Z}_N^+$ is selected as
\begin{subequations}
\label{eq:istarff}
     \begin{align}
                &i^{\bs{x}}_{kj} = \arg\max_{i \in \mathbb{Z}_N^+} \;
                \dot{h}^{kj}_{1,i} + \underline{J}^{kj}_{1,\bs{x}} \Delta_{\bs{x}_{i,I_{\dot{h}^{kj}_{1}}}}^+ - \overline{J}^{kj}_{1,\bs{x}} \Delta_{\bs{x}_{i,I_{\dot{h}^{kj}_{1}}}}^-,\\
                &\text{s.t.} \nonumber \\
                &\|(1-\lambda)\big(\bs{x}_{i^{\bs{x}}_{kj},\ell_k} - \bs{x}_{i^{\bs{x}}_{kj},\ell_j}\big)
                + \lambda\big(\bs{x}_{\ell_k} - \bs{x}_{\ell_j}\big)\|_2
                \geq \epsilon,
                \quad \forall\, \lambda \in [0,1],\label{ineq:ffdatacond}
    \end{align}
\end{subequations}
  where $\Delta_{\bs{x}_{i,I_{\dot{h}^{kj}_{1}}}} := \bs{x}_{I_{\dot{h}^{kj}_{1}}} - \bs{x}_{i,I_{\dot{h}^{kj}_{1}}}$. The condition \eqref{ineq:ffdatacond} ensures that the interpolation between the dataset point and the current states of agents $\ell_k$ and $\ell_j$ remains inside the state set $X$. This guarantees that the Jacobians of the barrier functions remain bounded along the interpolation segment, which is required for the application of Theorem~\ref{th:centralcbf}, which relies on the Mean Value Theorem along this segment to establish that the functions satisfy the 3D-ZCBF property.  In practice, this can be enforced by verifying the condition at the dataset point ($\lambda=0$) as well as at the critical point along the segment where the derivative of condition \eqref{ineq:ffdatacond} with respect to $\lambda$ vanishes.
  
The decoupled constraints can now be stated formally 
in the following lemma.
\begin{Lemma}\label{lemma:ffcbf}
     Let $e_{kj} \in E_{ff}$ where $k,j \in V_f$ be a follower–follower edge satisfying Assumption~\ref{as:ffleaders}. For the unknown control-affine multi-agent system~\eqref{agentZCBF},  consider the continuously differentiable functions $h^{kj}_{1},h^{kj}_{2}: X_{I_{h_r^{kj}}} \rightarrow \mathbb{R}$ as defined in \eqref{eq:cbfffparalel}, \eqref{eq:cbfffperp} with input-output datasets $\mathcal{D}^{kj}_1$, $\mathcal{D}^{kj}_2$ satisfying Assumptions~\ref{ass_data}.
    Let $i^{\bs{x}}_{kj} \in\mathbb{Z}_N^+$ denote the index defined in \eqref{eq:istarff}.
         Then $h^{kj}_{1}$ is a 3D-ZCBF for edge $e_{kj}$ 
    if there exist  a class $\mathcal{K}_{\infty}$ function $\alpha(\cdot)$ and constants $\beta_k, \beta_j \ge 0, \beta_k + \beta_j = 1$ such that
    \begin{subequations}
        \begin{align}
        &\sup_{\bs{u}_{\ell_k} \in U_{\ell_k}} \beta_k (\dot{h}^{kj}_{1,i^{\bs{x}}_{kj}} +\underline{J}^{kj}_{1,\bs{x}_{\mathcal{N}^+_{kj}}} \Delta^+_{\bs{x}_{i^{\bs{x}}_{kj},\mathcal{N}^+_{kj}}} - \overline{J}^{kj}_{1,\bs{x}_{\mathcal{N}^+_{kj}}} \Delta^-_{\bs{x}_{i^{\bs{x}}_{kj},\mathcal{N}^+_{kj}}} ) 
           \nonumber \\
           &+ \underline{J}^{kj}_{1,\bs{x}_k} \Delta^+_{\bs{x}_{i^{\bs{x}}_{kj},k}}
            - \overline{J}^{kj}_{1,\bs{x}_k} \Delta^-_{\bs{x}_{i^{\bs{x}}_{kj},k}} + \underline{J}^{kj}_{1,\bs{x}_{\ell_k}} \Delta^+_{\bs{x}_{i^{\bs{x}}_{kj},{\ell_k}}} 
            - \overline{J}^{kj}_{1,\bs{x}_{\ell_k}} \Delta^-_{\bs{x}_{i^{\bs{x}}_{kj},{\ell_k}}}  \nonumber\\  
            &+ \underline{J}^{kj}_{1,\bs{u}_{\ell_k}} \Delta^+_{\bs{u}_{i^{\bs{x}}_{kj},{\ell_k}}}  
            - \overline{J}^{kj}_{\bs{u}_{\ell_k}} \Delta^-_{\bs{u}_{i^{\bs{x}}_{kj},k}}  
            + \underline{J}^{kj}_{1,\bs{x}_{\mathcal{N}^+_{k \setminus j}}} \Delta^+_{\bs{x}_{i^{\bs{x}}_{kj},\mathcal{N}^+_{k \setminus j}}} \nonumber\\
            &- \overline{J}^{kj}_{1,\bs{x}_{\mathcal{N}^+_{k \setminus j}}} \Delta^-_{\bs{x}_{i^{\bs{x}}_{kj},\mathcal{N}^+_{k \setminus j}}} 
            \ge - \beta_k \, \alpha(h^{kj}_{1}(\bs{x}_{I_{h^{kj}_{1}}})), \forall\bs{x} \in X,   \\
            &\sup_{\bs{u}_{\ell_j} \in U_{\ell_j}} \beta_j (\dot{h}^{kj}_{1,i^{\bs{x}}_{kj}} +\underline{J}^{kj}_{1,\bs{x}_{\mathcal{N}^+_{kj}}} \Delta^+_{\bs{x}_{i^{\bs{x}}_{kj},\mathcal{N}^+_{kj}}} - \overline{J}^{kj}_{1,\bs{x}_{\mathcal{N}^+_{kj}}} \Delta^-_{\bs{x}_{i^{\bs{x}}_{kj},\mathcal{N}^+_{kj}}} ) 
            \nonumber \\
            &+ \underline{J}^{kj}_{1,\bs{x}_j} \Delta^+_{\bs{x}_{i^{\bs{x}}_{kj},j}}
            - \overline{J}^{kj}_{1,\bs{x}_j} \Delta^-_{\bs{x}_{i^{\bs{x}}_{kj},j}} + \underline{J}^{kj}_{1,\bs{x}_{\ell_j}} \Delta^+_{\bs{x}_{i^{\bs{x}}_{kj},{\ell_j}}} - \overline{J}^{kj}_{1,\bs{x}_{\ell_j}} \Delta^-_{\bs{x}_{i^{\bs{x}}_{kj},{\ell_j}}}  \nonumber\\ 
            &+ \underline{J}^{kj}_{1,\bs{u}_{\ell_j}} \Delta^+_{\bs{u}_{i^{\bs{x}}_{kj},{\ell_j}}} 
            - \overline{J}^{kj}_{1,\bs{u}_{\ell_j}} \Delta^-_{\bs{u}_{i^{\bs{x}}_{kj},j}}  
            + \underline{J}^{kj}_{1,\bs{x}_{\mathcal{N}^+_{j \setminus k}}} \Delta^+_{\bs{x}_{i^{\bs{x}}_{kj},\bs{x}_{\mathcal{N}^+_{j \setminus k}}}} \nonumber \\
            &- \overline{J}^{kj}_{1,\bs{x}_{\mathcal{N}^+_{j \setminus k}}} \Delta^-_{\bs{x}_{i^{\bs{x}}_{kj},\mathcal{N}^+_{j \setminus k}}}   
            \ge - \beta_j \, \alpha(h^{kj}_{1}(\bs{x}_{I_{h^{kj}_{1}}})), \forall\bs{x} \in X,\\
        \end{align}
        \end{subequations}
where $\Delta_{\bs{x}_{i^{\bs{x}}_{kj},A}} := \bs{x}_A - \bs{x}_{i^{\bs{x}}_{kj},A}, 
    \Delta_{\bs{u}_{i^{\bs{x}}_{kj},A}} := \bs{u}_A - \bs{u}_{i^{\bs{x}}_{kj},A}$, 
    for any agent index set  
    $A \subseteq \{k,j\} \cup \mathcal{N}^+_{kj} \cup \mathcal{N}^+_{k \setminus j} \cup \mathcal{N}^+_{j \setminus k}$ and  $\mathcal{N}^+_{kj} := (\mathcal{N}_k \cup \mathcal{N}_{\ell_k}) \cap (\mathcal{N}_j \cup \mathcal{N}_{\ell_j})$ denotes the neighbors shared by agents $k$ and $j$ including their leaders, while  
   $ \mathcal{N}^+_{k \setminus j} := (\mathcal{N}_k \cup \mathcal{N}_{\ell_k}) \setminus (\mathcal{N}_j \cup \mathcal{N}_{\ell_j}), \quad 
    \mathcal{N}^+_{j \setminus k} := (\mathcal{N}_j \cup \mathcal{N}_{\ell_j}) \setminus (\mathcal{N}_k \cup \mathcal{N}_{\ell_k})$
    denote the neighbors of $k$ and $j$ (and their leaders) that are not shared with the other agent. The same statement holds for the orthogonal function $h^{kj}_2$, i.e., $h^{kj}_2$ is a 3D-ZCBF for edge $e_{kj}$ if the optimization problem above remains feasible when $h^{kj}_1$ is replaced by $h^{kj}_2$ in all terms.
\end{Lemma}
     \begin{remark}
            It should be noted that Assumption~\ref{as:ffleaders} is made primarily to simplify the definition of the ZCBF candidates in \eqref{eq:cbfffparalel}–\eqref{eq:cbfffperp}. In general, only one leader state is required for the barrier function to include a corresponding input in its time derivative, which is the minimal requirement for the 3D-ZCBF condition to be verifiable from first-order input–output data. If only one follower e.g. $k$ has a leader, the parallel projection in the ZCBF can be taken along the vector connecting that leader and the follower $\bs{x}_{\ell_k}-\bs{x}_k$, instead of a vector connecting two distinct leaders $\bs{x}_{\ell_k}-\bs{x}_{\ell_j}$. Therefore, the assumption of two distinct leaders is not strictly necessary for the approach to work; it is introduced for notational convenience and to avoid defining multiple separate cases
        \end{remark}
  
\subsection{Distributed Safety Synthesis and Implementation}

Overall system safety can be guaranteed when the 3D-ZCBF conditions hold on every edge: leader–follower edges satisfying \eqref{eq:cbfdd_cond}, leader–leader edges satisfying Lemma~\ref{lemma:llcbf}, and follower–follower edges satisfying Lemma~\ref{lemma:ffcbf}. Under these conditions, network connectivity is maintained, as shown in the following theorem.
\begin{Theorem}\label{prop:robust-invariance}
        Consider the unknown control-affine multi-agent system \eqref{agentZCBF} with input–output datasets $\mathcal{D}^{kj}_{r}, e_{kj}\in E, \quad r=1,\cdots,n_{kj}$ satisfying Assumption~\ref{ass_data}. For each edge $e_{kj}\in E$, let
        \begin{align*}
             h^{kj}_{1} := 
        \begin{cases}
        \text{\eqref{eq:cbflf}}, & e_{kj} \in E_{lf},\\
        \text{\eqref{eq:cbfll}}, & e_{kj} \in E_{ll},\\
        \text{\eqref{eq:cbfffparalel}}, & e_{kj} \in E_{ff},
        \end{cases}
        \qquad
        h^{kj}_{2} := \text{\eqref{eq:cbfffperp}}, \quad e_{kj} \in E_{ff}.
        \end{align*}
     Then $h^{kj}_{1}$ and $h^{kj}_{2}$ are 3D-CBFs if, for each edge, the corresponding condition holds: \eqref{eq:cbfdd_cond} for $e_{kj}\in E_{lf}$, Lemma~\ref{lemma:llcbf} for $e_{kj}\in E_{ll}$, and Lemma~\ref{lemma:ffcbf} for $e_{kj}\in E_{ff}$.
   Moreover, let $ S_{1}^{kj}$, $ S_{2}^{kj}$, be the corresponding safety sets of $h^{kj}_{1}$, $h^{kj}_{2}$ for $e_{kj} \in E$ and let $i^{\bs{x}}_{kj}$ be defined as in \eqref{eq:istarlf}  for $e_{kj}\in E_{lf}$, \eqref{eq:istarll} for $e_{kj}\in E_{ll}$ and \eqref{eq:istarff} for $e_{kj}\in E_{ff}$. Then,  any Lipschitz continuous controller
    \begin{align*}
       &u(\bs{x}) \in K_S(\bs{x})=
        \bigcap_{e_{kj}\in E} K_{S_{1}^{kj},i^{\bs{x}}_{kj}}
        \cap 
        \bigcap_{e_{kj}\in E_{ff}} K_{S_{2}^{kj},i^{\bs{x}}_{kj}}  
    \end{align*}
    for the system~\eqref{agentZCBF} renders the set
    \begin{align*}
        &S= \bigcap_{e_{kj}\in E} S_{1}^{kj}
        \cap 
        \bigcap_{e_{kj}\in E_{ff}} S_{2}^{kj}
    \end{align*}
    controlled invariant. Therefore, a maximum distance of $d_{max}$ is kept between neighbouring agents, i.e. $\|\bs{x}_k-\bs{x}_j\|_2\le d_{max}, \forall e_{kj}\in E, \forall t\ge0$. 
    \end{Theorem}
\begin{proof}
   By Theorem~\ref{th:centralcbf}, any controller chosen from the individual safe input sets renders the corresponding safety set controlled invariant. Controlled invariance is preserved under intersection, so a controller that belongs to all safe input sets simultaneously ensures that the intersection of all safety sets is also controlled invariant. By construction of $h_1^{kj}$ and $h_2^{kj}$, this implies $\|\bs{x}_k-\bs{x}_j\|_2\le d_{\max}$.
\end{proof}
This theorem shows that the overall system can be rendered safe with ZCBF conditions decoupled across individual inputs. To realize this in practice, one may modify a nominal controller minimally so that the safety constraints are satisfied, leading to the following optimization-based control law.

\paragraph{Practical controller implementation}  In practice, a nominal input $u^{\mathrm{nom}}(\bs{x})$ may not satisfy all safety constraints. A safe input can be computed by projecting $u^{\mathrm{nom}}(\bs{x})$ onto the intersection of the safe sets $K_S(\bs{x})$. To ensure that this projection always has a solution, a non-negative slack variable $s \ge 0$ can be introduced, leading to the following optimization problem.
\begin{subequations}\label{eq:cbfOptProblem}
     \begin{align}
        u^*(\bs{x})
        &= \arg\min_{u(\bs{x}),\,s\ge 0} 
        \; \tfrac{1}{2} \|u(\bs{x}) - u^{\mathrm{nom}}(\bs{x})\|_2^2 + \rho\, \|s\|_1,\\
        \text{s.t.} \quad & u(\bs{x}) \in K_S(\bs{x}) + s,
    \end{align}
\end{subequations}
where $\rho>0$ penalizes slack violation.

\section{Simulation Results}
We evaluate the method on leader–follower formations \footnote{The code is available at \url{https://github.com/mirhanu/3D-ZCBF}.} with dynamics
\begin{subequations}
    \begin{align}
        &\dot{\bs{x}}_i = -\sum_{j \in \mathcal{N}_i} \left(\bs{x}_i - \bs{x}_j - d_{ij}^{\mathrm{des}}\right), \quad i \in V_f,\\
        &\dot{\bs{x}}_i = -\sum_{j \in \mathcal{N}_i} \left(\bs{x}_i - \bs{x}_j - d_{ij}^{\mathrm{des}}\right)+\bs{u}_i, \quad i \in V_l,
    \label{eq:follower_dynamics}
    \end{align}
\end{subequations}
where $d_{ij}^{\mathrm{des}}$ is the desired relative displacement. The dynamics are designed so that followers naturally move toward $d{ij}^{\mathrm{des}}$, while leaders can additionally use $\bs{u}_i$ to achieve higher-level objectives such as target tracking. We set $d_{max}=3$ and $\beta_k=\beta_j=0.5$ for all leader–leader and follower–follower edges.

Two instances with different graphs $G$ are studied: one with 1-Dimensional states and one with 2-Dimensional states. Datasets $\mathcal{D}_r^{kj}$ are generated from random initial conditions, with each agent’s state and each leader’s input drawn i.i.d.\ from $[-5,5]$. For the 1-D case, 50 simulations run for $1\,\mathrm{s}$ with $0.01\,\mathrm{s}$ sampling; for the 2-D case, 50 simulations run for $0.01\,\mathrm{s}$ with $0.001\,\mathrm{s}$ sampling, using a smaller step size due to the sensitivity of $\dot h_1^{23}$. The derivatives $\dot{h}_r^{kj}$ are obtained via central finite differences from $h_r^{kj}$, keeping only samples that satisfy \eqref{eq:ZCBFcond}. Upper and lower bounds $\underline{J}_r^{kj}$ and $\overline{J}_r^{kj}$ are estimated following \cite{jin2022data} by solving a bound-constrained optimization over all sample pairs. Because the number of constraints grows quadratically with the dataset size (i.e., $O(|\mathcal{D}|^{2})$), applying the method to all samples would lead to an intractably large problem. To reduce complexity, we cluster the data using $k$-means and retain only the points closest to each cluster center as representative samples. For each case, \eqref{eq:cbfOptProblem} is formulated in \texttt{cvxpy} and solved with OSQP. Leaders use the nominal P controller
\begin{equation}
\bs{u}^{\mathrm{nom}}_i(\bs{x})=-k_p\big(\bs{x}_i-d_i^{\mathrm{des}}\big), \qquad i\in V_l,
\end{equation}
where $d_i^{\mathrm{des}}$ is the desired leader position.
\subsection{Case A: 1D Leader–Follower}\label{sec:1d}
We consider a one-dimensional leader–follower network with four agents and communication graph
$G=(V,E,A)$, where $V=\{0,1,2,3\}$, $E=\{(0,1),(0,2),(2,3),(0,3)\}$, with leaders
$V_l=\{0,3\}$ and followers $V_f=\{1,2\}$. All edges have a desired relative displacement of $d_{ij}^{\mathrm{des}} = 1$, and the leaders’ individual targets are $d_0^{\mathrm{des}} = 1$ and $d_3^{\mathrm{des}} = 5$. The nominal proportional gain is set to $k_p = 15$, which would make the system unsafe without the CBF framework.

Simulations start from the initial state $\bs{x} = [-0.5,\,-1.0,\,1.5,\,2.0]^\top$. As shown in Fig.~\ref{fig:1_results}, safety is maintained with ZCBF values remaining strictly positive. A safety margin is kept due to the robust formulation with upper/lower Jacobian bounds. A chattering in the control inputs is observed, just as reported in \cite{jin2023robust}. This is likely caused by the nonsmoothness of the finite data set. The effect tends to decrease as the data set size increases.




\begin{figure}[thpb]
  \centering
  \begin{subfigure}{\linewidth}
    \centering
    \includegraphics[width=\linewidth]{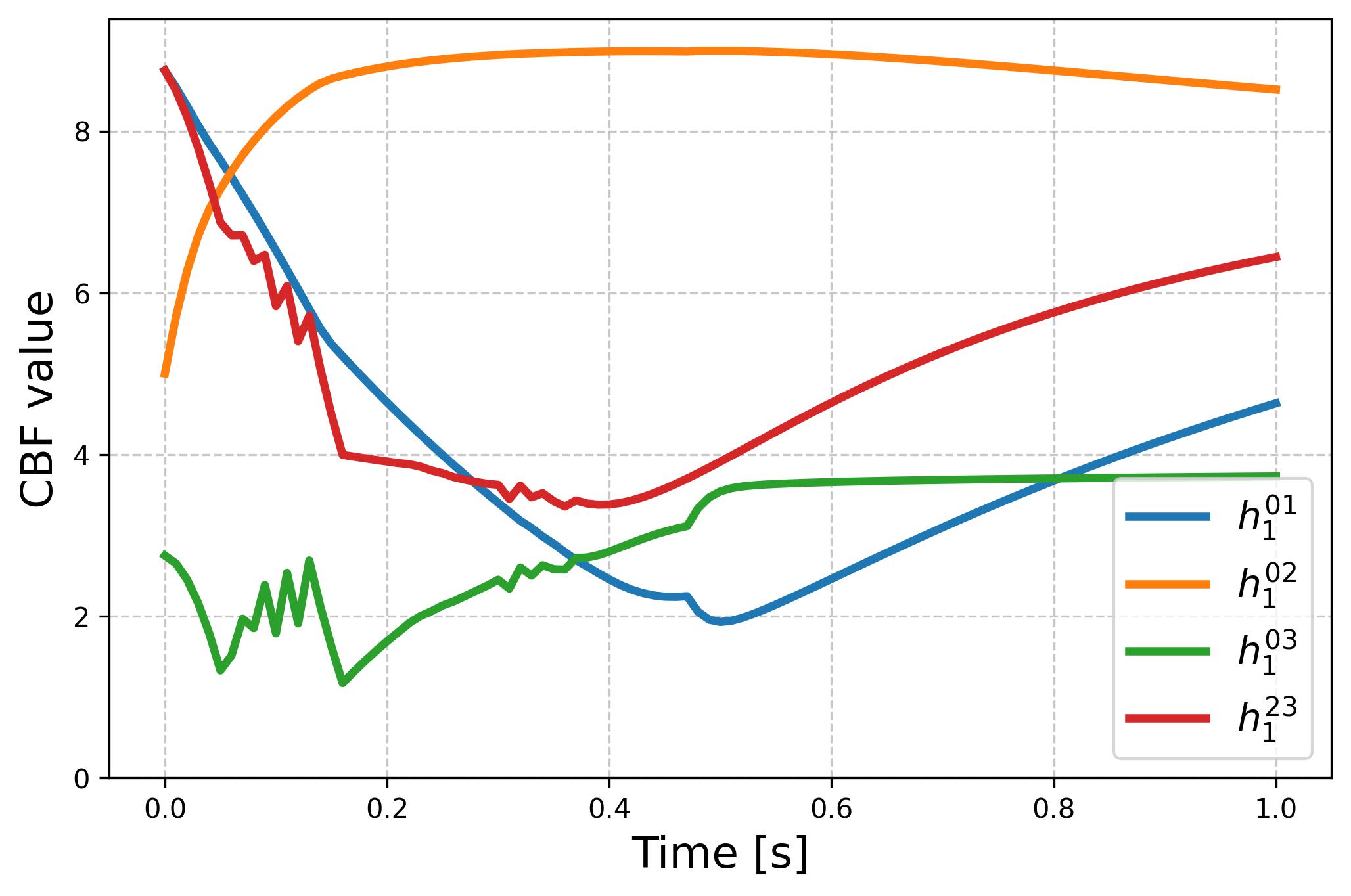}
    \caption{1-Dimensional case: ZCBF values $h^{kj}_r(t)$ over time}
  \end{subfigure}

  \vspace{0.3cm}

  \begin{subfigure}{\linewidth}
    \centering
    \hspace{-0.3cm}\includegraphics[width=1.02\linewidth]{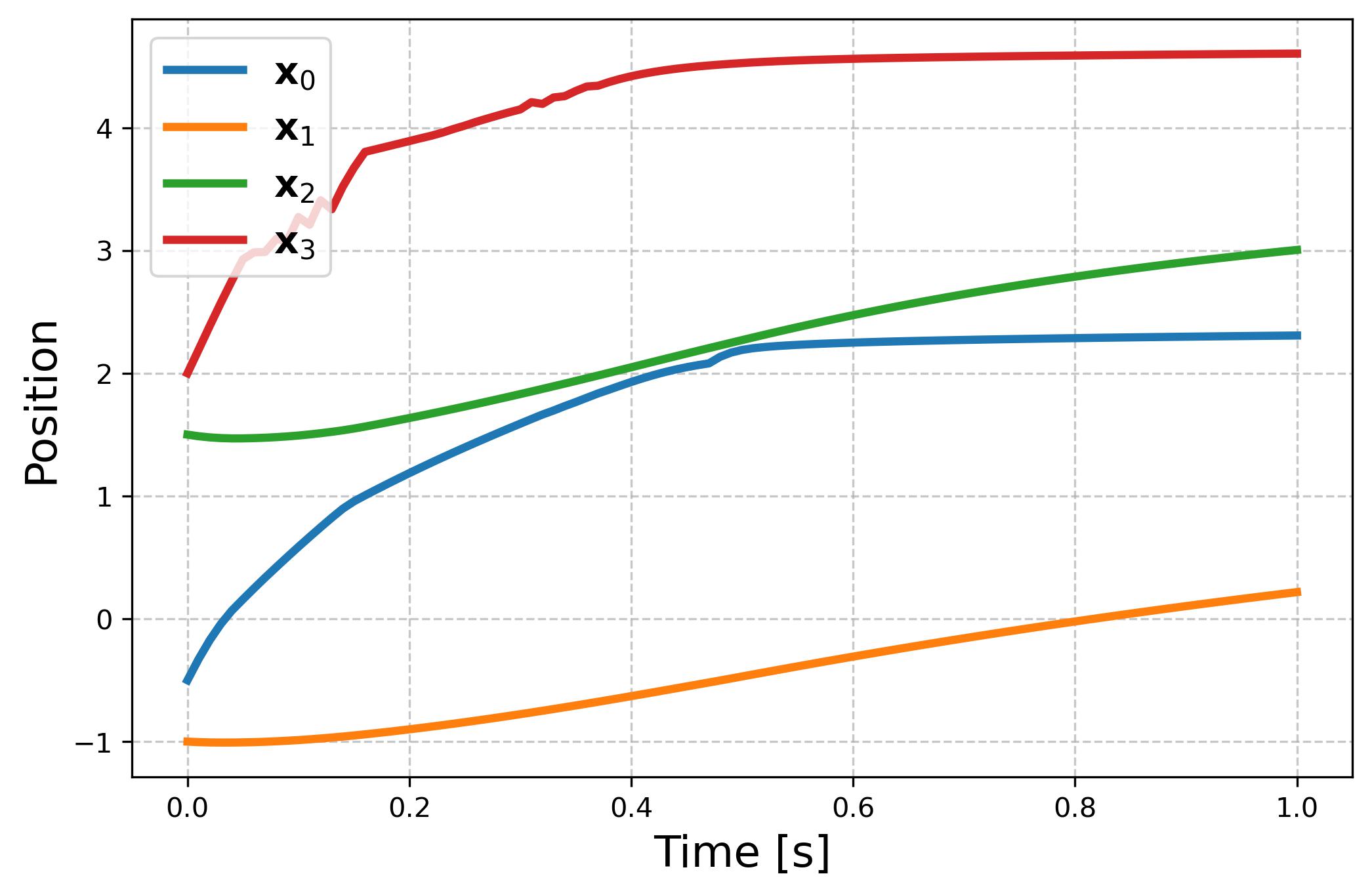}
    \caption{1-Dimensional case: Agent states and trajectories}
  \end{subfigure}

  \vspace{0.3cm}

  \begin{subfigure}{\linewidth}
    \centering
    \hspace{-0.6cm}\includegraphics[width=1.05\linewidth, height=0.67\linewidth]{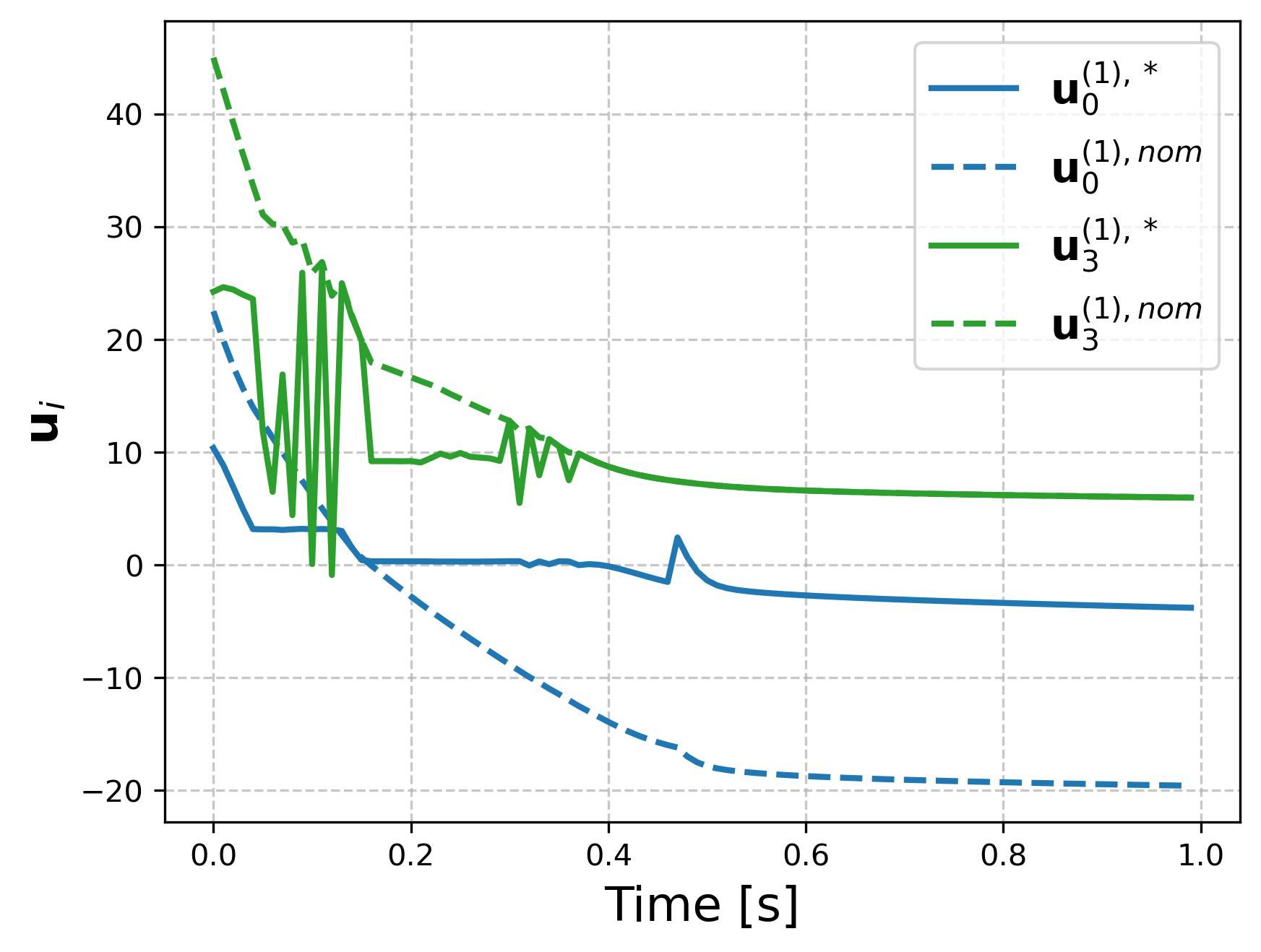}
    \caption{1-Dimensional case: Control inputs}
  \end{subfigure}

  \caption{Simulation results for 1-Dimensional leader–follower system with connectivity preservation.  Right column: 2-Dimensional system. Top row: ZCBF values $h^{kj}_r(t)$ over time, confirming forward invariance of the safe sets. Middle row: Agent states and trajectories (for 2D case, circles = start points, squares = end points). Bottom row: Control inputs: nominal (dashed) vs. 3D-ZCBF-filtered (solid), showing minimal deviation while enforcing safety. Superscripts $(1)$ and $(2)$ denote the first and
second input components, respectively.}
  \label{fig:1_results}
\end{figure}

\subsection{Case B: 2D Leader–Follower}
We consider a planar leader–follower system with four agents and a communication graph
$G=(V,E,A)$, where $V=\{0,1,2,3\}$, $E=\{(0,1),(1,2),(2,3)\}$, with leaders
$0,3\in V_l$ and followers $1,2\in V_f$. Desired relative positions are
$d_{01}^{\mathrm{des}}=[1,\,2]^\top$, $d_{12}^{\mathrm{des}}=[2,\,1]^\top$, and
$d_{23}^{\mathrm{des}}=[1,\,1]^\top$. Leader desired pose are
$d_{0}^{\mathrm{des}}=[1,\,1]^\top$ and $d_{3}^{\mathrm{des}}=[5,\,5]^\top$,
with a nominal P gain $k_p=10$. The system is simulated from initial conditions
$\bs{x}(0)=[0.0,\,0.0,\,0.5,\,0.5,\,1.0,\,1.0,\,1.5,\,1.5]^\top$.
Figure~\ref{fig:2_results} shows that safety is maintained, as all ZCBF values
remain positive. The control input $\bs{u}_3^*$ deviates from the nominal
$\bs{u}^{\mathrm{nom}}_3$ during the first $2\,\mathrm{s}$ to enforce
$h_1^{23}>0$.

\begin{figure}[thpb]
  \centering
  \begin{subfigure}{\linewidth}
    \centering
    \includegraphics[width=0.98\linewidth]{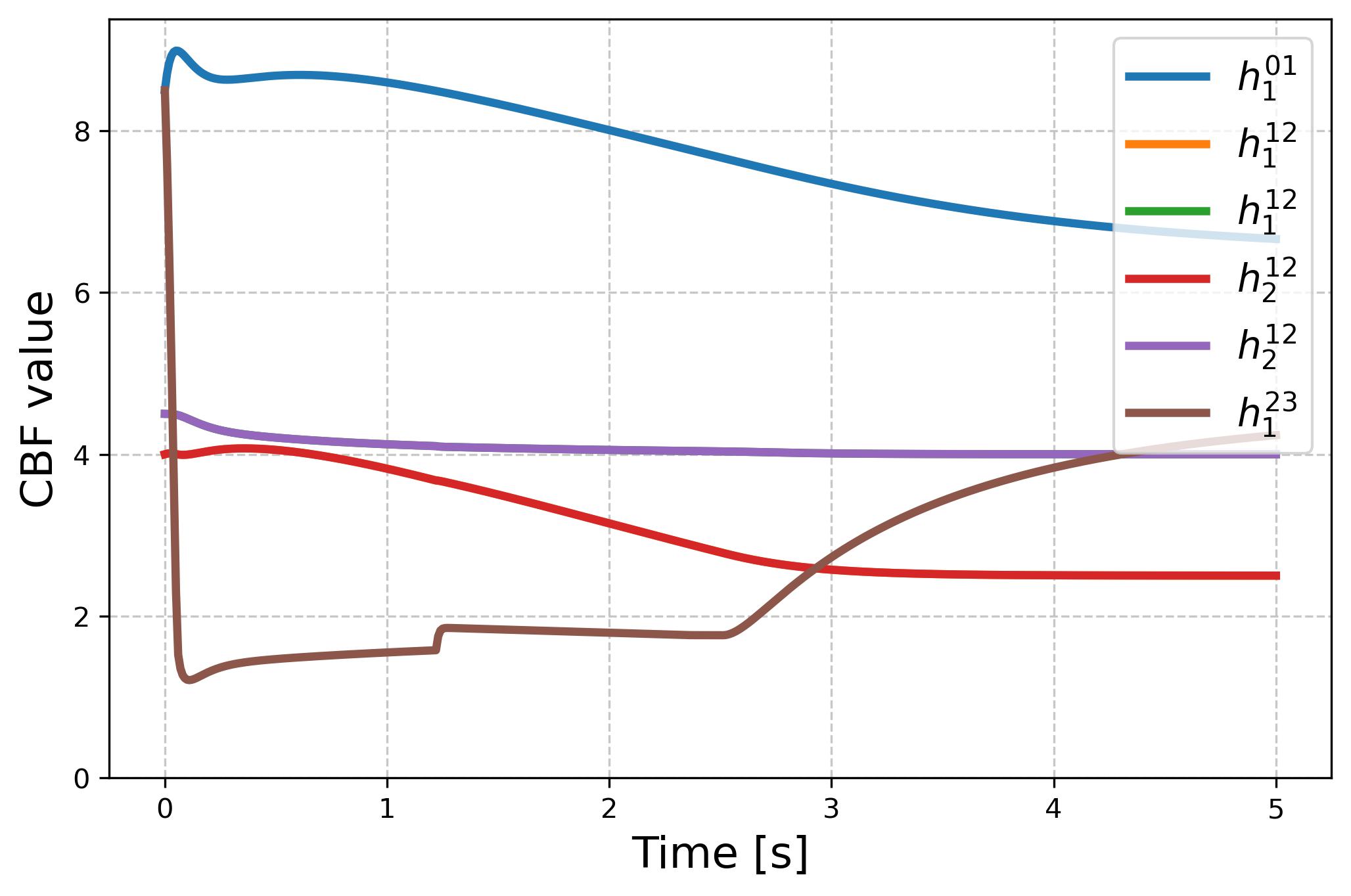}
    \caption{2-Dimensional case: ZCBF values $h^{kj}_r(t)$ over time}
  \end{subfigure}

  \vspace{0.3cm}

  \begin{subfigure}{\linewidth}
    \centering
    \hspace{-0.3cm}\includegraphics[width=1.02\linewidth, height=0.67\linewidth]{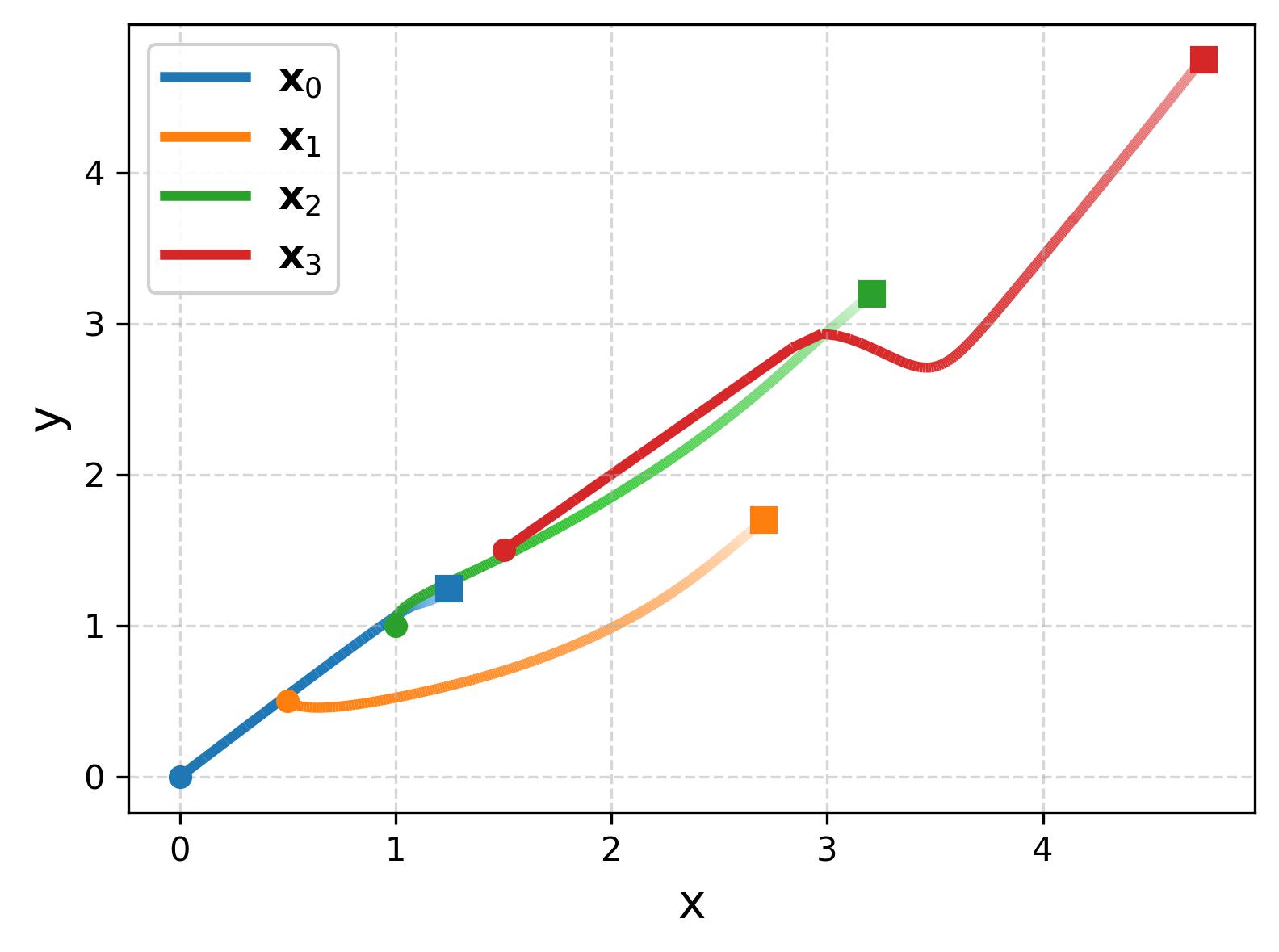}
    \caption{2-Dimensional case: Agent states and trajectories}
  \end{subfigure}

  \vspace{0.3cm}

  \begin{subfigure}{\linewidth}
    \centering
    \hspace{-0.5cm}\includegraphics[width=1.03\linewidth, height=0.67\linewidth]{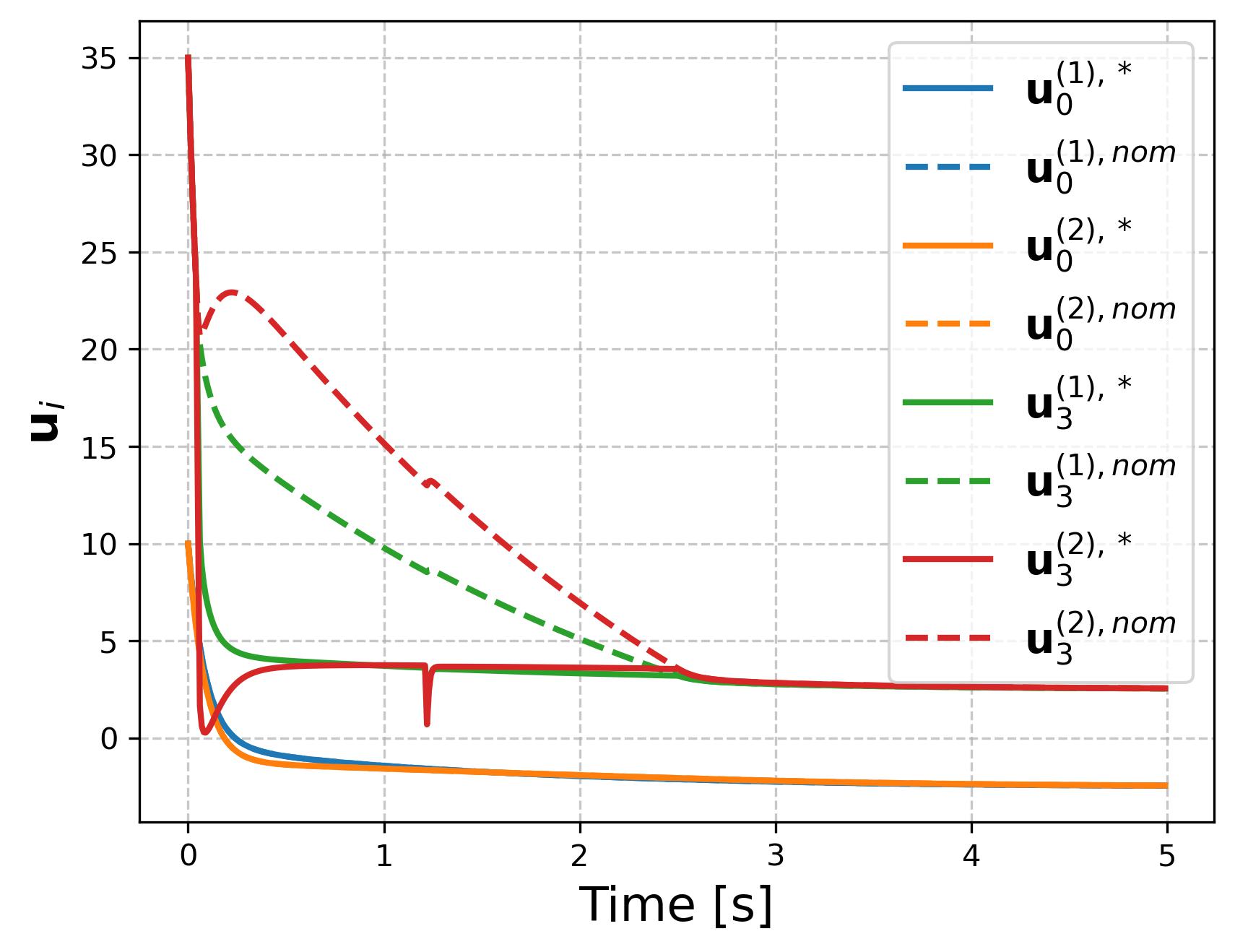}
    \caption{2-Dimensional case: Control inputs}
  \end{subfigure}

  \caption{Simulation results for 2-Dimensional leader–follower system with connectivity preservation. Top: ZCBF values $h^{kj}_r(t)$ over time, confirming forward invariance of the safe sets. Middle: Agent states and trajectories (for 2D case, circles = start points, squares = end points). Bottom: Control inputs: nominal (dashed) vs. 3D-ZCBF-filtered (solid), showing minimal deviation while enforcing safety. Superscripts $(1)$ and $(2)$ denote the first and
second input components, respectively.}
  \label{fig:2_results}
\end{figure}

\subsection{Effect of Dataset Size and Jacobian Bounds}

We investigate the effect of the amount of data used for Jacobian estimation and the choice of Jacobian bounds on the performance and safety of the 3D-ZCBF controller, using the 1-D leader–follower system with the same parameters as in Section~\ref{sec:1d}. Specifically, we vary the number of simulations used to generate the datasets $\mathcal{D}_r^{kj}$, which in turn determines the dataset size, and we additionally scale the estimated Jacobian bounds. For each configuration, 10 closed-loop simulations are run from the same initial condition, each using a different dataset generated from independent simulations. The datasets $\mathcal{D}_r^{kj}$ are generated from random initial conditions, with each agent’s state and each leader’s input drawn i.i.d.\ from $[-5,5]$ for each simulation. Performance and safety are assessed using three metrics: (i) the average control cost across the 10 simulations, defined as half the squared 2-norm of the deviation from the nominal controller (see~\eqref{eq:cbfOptProblem}); (ii) the minimum ZCBF value $h_r^{kj}(t)$ attained over all agents after 0.1 seconds, averaged across the simulations when no constraint violations occur; and (iii) the number of time instants at which any ZCBF constraint is violated, i.e., $h_r^{kj}(t)<0$, with the number of simulations in which violations occur reported in parentheses. We exclude the first 0.1 seconds from the minimum ZCBF calculation because initial conditions can result in temporarily low ZCBF values that do not reflect the safety margin maintained during normal operation.

\begin{table}[htbp]
\centering
\caption{Average simulation results for different numbers of simulations used to generate the datasets and for variations in Jacobian bounds. The $ 0.5 \times$ and $ 2 \times$  indicate Jacobian bounds that are learned from data scaled by half and double, respectively. Reported are the average control cost (half the squared 2-norm deviation from the nominal controller) and the average minimum ZCBF value $h_r^{kj}$ across all agents  (excluding the first 0.1 seconds), and the total number of constraint violations  ($h_r^{kj}(t)<0$) out of $1000$ instances accross all simulations, with the number of simulations in which violations occurred indicated in parentheses.}
\label{tab:average_results}
\begin{tabular}{lccc}
\toprule
Configuration & Control Cost & Min   $h_r^{kj}$ &  Violations \\
\midrule
5 sims & $8.54$ & $2.11$ & 59 (1) \\
50 sims & $6.49$ & $0.99$ & 34 (5)\\
500 sims & $6.32$ & $0.17$ & 56 (7)\\
50 sims, $0.5 \times$ bounds & $0.77$ & - & 734 (10) \\
50 sims, $2 \times$ bounds & $7.21$ & $2.96$ & 0 (0) \\
500 sims, $2 \times$ bounds & $6.85$ & $3.02$ & 0 (0) \\
\bottomrule
\end{tabular}
\end{table}

The results in Table~\ref{tab:average_results} show that increasing the dataset size reduces control cost and reduces the minimum ZCBF value $h_r^{kj}$ attained, meaning the controller becomes less conservative with increasing dataset as expected, since larger data covers more space and gives better estimates of $\dot{h}_r^{kj}$.

The total number of ZCBF constraint violations does not increase monotonically with the data set size. In particular, violations are observed even for the smallest dataset, constructed from 5 simulations, and the total number of violations for datasets generated from 5, 50, and 500 simulations remains in similar amountts. However, the frequency across simulations increases with dataset size: violations occur in only 1 out of 10 simulations when the dataset is generated from 5 simulations, compared to 5 and 7 simulations when the datasets are generated from 50 and 500 simulations, respectively. This indicates that while the overall number of violations does not grow with increasing data, violations occur more consistently as the controller becomes less conservative with larger datasets.

The single simulation exhibiting violations for the dataset generated from 5 simulations can be attributed to limited data availability.  With such a small dataset, the available data may provide insufficient coverage around the current state, rendering the ZCBF constraints infeasible and preventing the controller from identifying any admissible control input. As a result, the system state remains in the unsafe region for an extended period, which leads to a large number of constraint violations within that simulation.

The choice of Jacobian bounds significantly affects controller conservatism. Halving the learned bounds results in insufficient safety margins, leading to frequent constraint violations and negative ZCBF values. In contrast, doubling the bounds yields a more conservative controller, increasing both the control cost and the minimum attained ZCBF values, while fully eliminating constraint violations.

We further investigated the effect of dataset size on the chattering of the control inputs. Figure~\ref{fig:smooth} shows representative input trajectories obtained using datasets $\mathcal{D}_r^{kj}$ generated from 5, 50, and 500 data-collection simulations. As the dataset size increases, the chattering in the control input is progressively reduced, and the trajectories become smoother, reflecting improved controller performance with more comprehensive data coverage.

\begin{figure}[thpb]
  \centering
  \begin{subfigure}{\linewidth}
    \centering
    \includegraphics[width=0.98\linewidth]{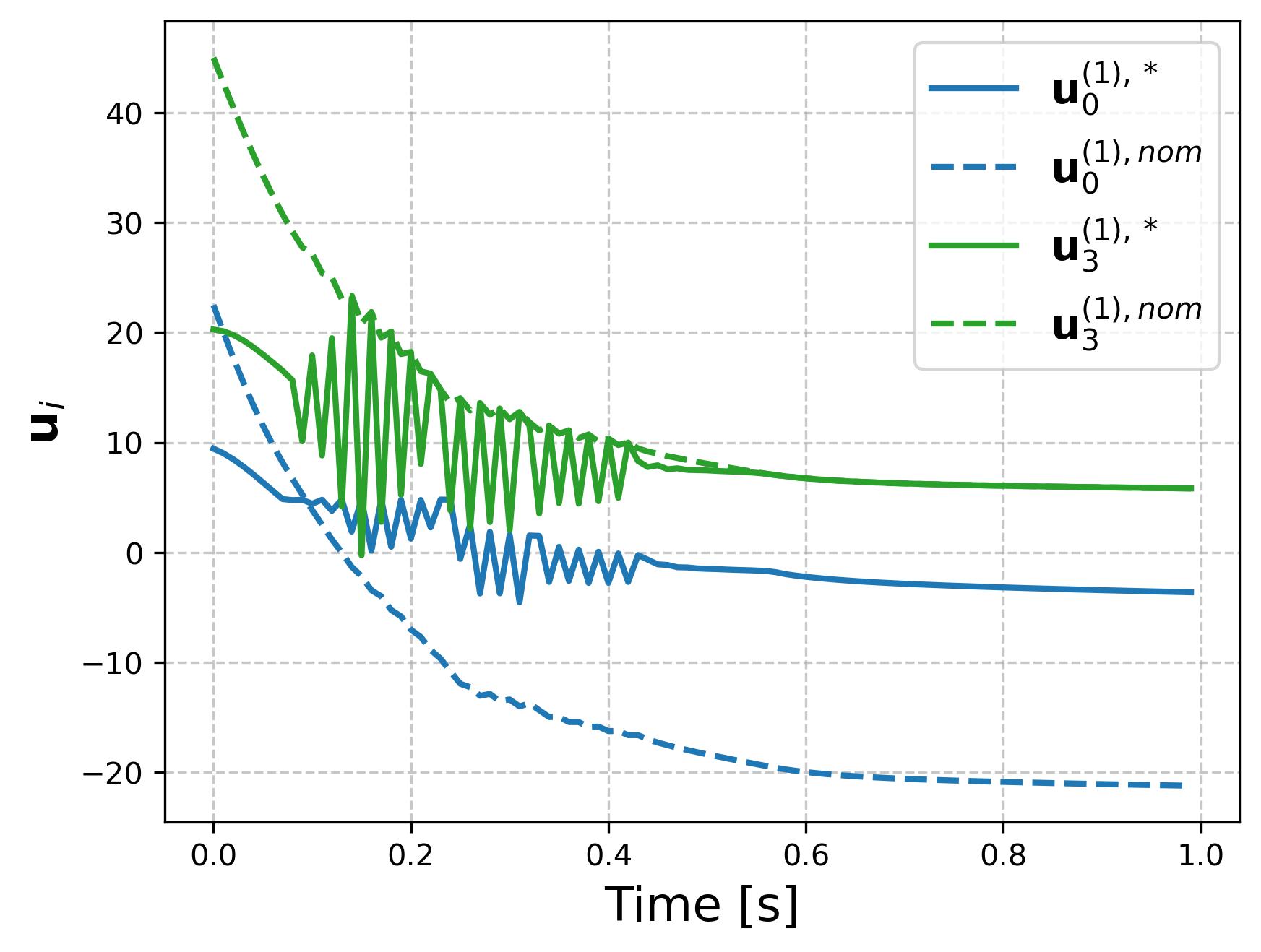}
    \caption{Control inputs using a dataset generated from 5 simulations}
  \end{subfigure}

  \vspace{0.3cm}

  \begin{subfigure}{\linewidth}
    \centering
    \hspace{-0.3cm}\includegraphics[width=1.02\linewidth, height=0.67\linewidth]{figures/control_edges_01-02-23-03_d_1-1-1-1.jpeg}
    \caption{Control inputs using a dataset generated from 50 simulations}
  \end{subfigure}

  \vspace{0.3cm}

  \begin{subfigure}{\linewidth}
    \centering
    \hspace{-0.5cm}\includegraphics[width=1.03\linewidth, height=0.67\linewidth]{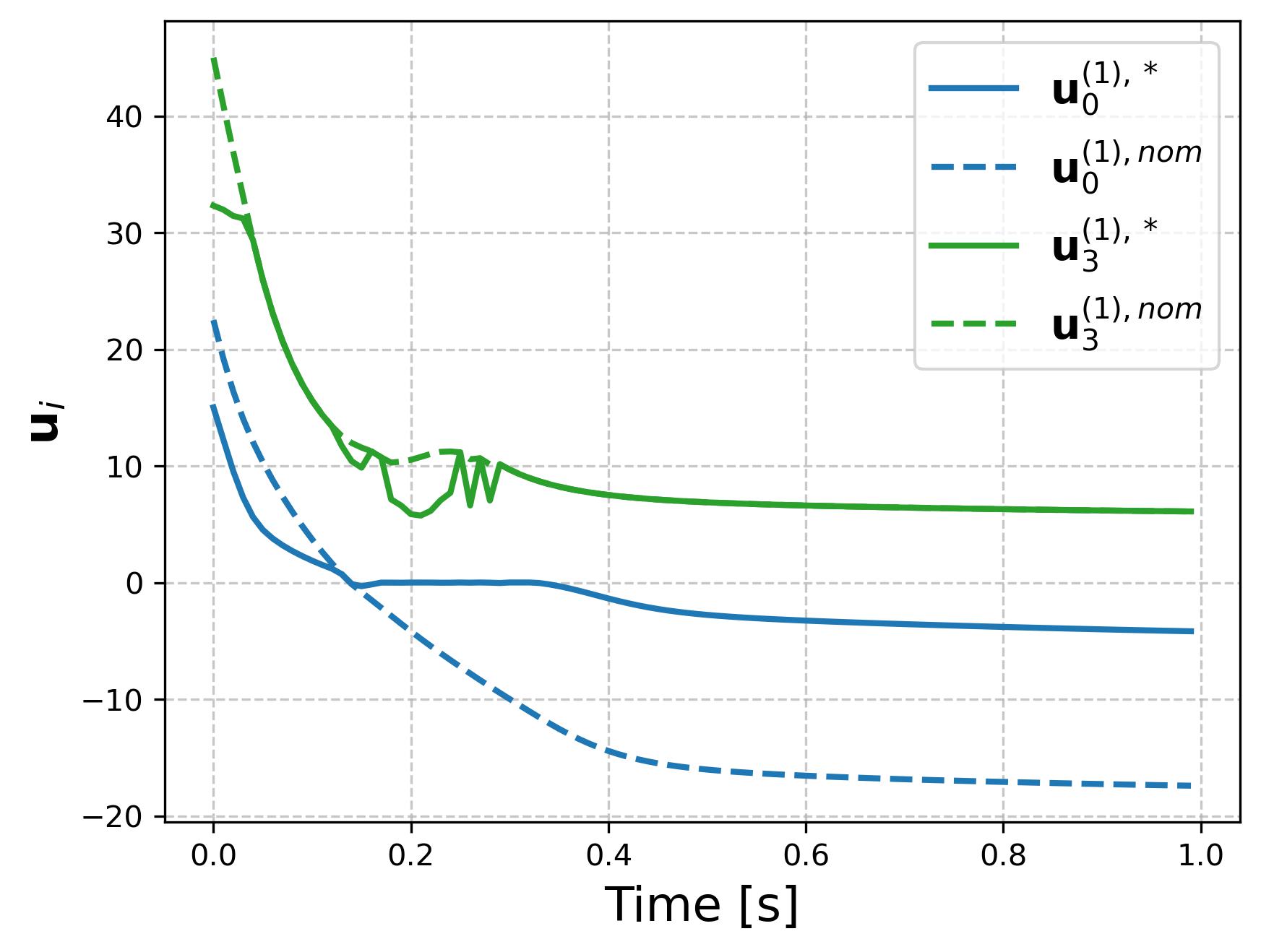}
    \caption{Control inputs using a dataset generated from 500 simulations}
  \end{subfigure}
  \caption{Representative control input trajectories for datasets generated from 5, 50, and 500 simulations. As the dataset size increases, the control inputs become  smoother and exhibit reduced chattering.}
  \label{fig:smooth}
\end{figure}

\section{Conclusion}\label{sec:conclusion}
This paper presented a complete framework 3D-ZCBFs to maintain connectivity in leader–follower multi-agent systems with unknown dynamics. By identifying bounds on CBF time derivatives directly from input–state data, the framework ensures the controlled invariance of safety sets without requiring explicit system models. We derived the explicit, decoupled safety conditions necessary for maintaining connectivity across leader–leader, and follower–follower pairings, and established system-wide conditions to guarantee network-level connectivity. Furthermore, we provided a quantitative analysis of how data set size and the resulting tightness of Jacobian bounds affect the feasibility and performance of the safety certificates. Simulation results evaluated these conditions and characterized the  effects of data set sizes and  Jacobian bounds. The results showed that while larger datasets reduce control cost, safety violations can still occur due to the inherent inaccuracy of learned bounds. We demonstrated that doubling the Jacobian bounds can eliminate these violations, albeit at the cost of increased controller conservatism. These findings highlight that the reliability of data-driven safety certificates is fundamentally tied to the robustness of the learned derivative bounds. Future work will focus on extending this data-driven framework to cases where the CBFs themselves are not known explicitly.
\bibliographystyle{IEEEtran}
\bibliography{Ref.bib}

\end{document}

%% file: figures/ff.tikz
\begin{tikzpicture}[
  >=Stealth,
  scale=1.2,
  show background rectangle,
  background rectangle/.style={fill=gray!5, rounded corners=2mm}
]

\tikzset{
    leader/.style={circle, draw, thick, fill=gray!25, minimum size=6mm},
    follower/.style={circle, draw, thick, fill=blue!20, minimum size=6mm},
    projection/.style={circle, draw, fill=gray!20, minimum size=1mm, inner sep=0.2mm},
vec/.style={->, line width=1.1pt},
para/.style={->, dashed, line width=0.9pt},
perp/.style={->, dotted, line width=0.9pt},
}

\node[leader] (lk) at (0,-1) {$\ell_k$}; 
\node[leader] (lj) at (5,-1) {$\ell_j$};

\node[follower] (k) at (1.5,3.2) {$k$};
\node[follower] (j) at (2.8,1.9) {$j$};

\pgfmathsetmacro{\dmax}{2.9} 

\draw[vec] (lj) -- (lk)
    node[midway, below] {$\bs{\bar x}_{\ell_k\ell_j}$};

\draw[vec] (j) -- (k)
    node[midway, right] {$\bs{\bar x}_{kj}$};

\coordinate (kj) at ($(k)-(j)$);
\node[projection] (p) at ($(j)+(lj)!(k)!(lk)-(lj)!(j)!(lk)$) {};

\draw[para] (j) -- (p)
    node[midway, below right] {$\bs{\bar x}_{kj}^{\parallel}$};

\draw[perp] (p) -- (k)
    node[midway, left] {$\bs{\bar x}_{kj}^{\perp}$};

\draw ($(p)+(0.15,0)$) -- ++(0,0.15) -- ++(-0.15,0);
\filldraw[black] ($(p)+(0.07,0.07)$) circle (0.03);

\draw[dashed, gray] (p) circle ({\dmax/sqrt(2)});

\draw[->, gray] (p) -- ++({-\dmax/sqrt(2)*cos(-10)},{\dmax/sqrt(2)*sin(-10)}) 
    node[midway, below left] {$d_{\max}/\sqrt{2}$};

\coordinate (kperpmax) at ($(p)!1.6!(k)$);
\draw[->, thick, red!60!black] (k) -- (kperpmax)
    node[midway, right, font=\small] {$h^{kj}_2= \frac{d_{\max}^2}{2} - \|\bar x_{kj}^{\perp}\|_2^2$};

\coordinate (kparmax) at ($(p)!1.6!(j)$);
\draw[->, thick, red!60!black] (j) -- (kparmax)
    node[midway, right, xshift=2mm, font=\small] 
    {$h^{kj}_1 = \frac{d_{\max}^2}{2} - \|\bar x_{kj}^{\parallel}\|_2^2$};

\coordinate (proj) at ($(lk)!(p)!(lj)$);

\draw[perp, thin, gray!90] (k) -- ($(lk)!(p)!(lj)$);

\draw[gray!90, thin] ($(proj)+(0.15,0)$) -- ++(0,0.15) -- ++(-0.15,0);
\filldraw[gray!90] ($(proj)+(0.07,0.07)$) circle (0.03);

\end{tikzpicture}